\documentclass[11pt,a4paper]{article}
\usepackage[left=15mm, right=15mm, top=20mm,bottom=25mm]{geometry}
\usepackage{amsmath}
\usepackage{amsthm}
\usepackage{ mathrsfs }
\usepackage{amssymb}
\usepackage{xcolor}
\usepackage{breqn}
\usepackage{setspace}

\newtheorem{theorem}{Theorem}

\newtheorem{proposition}[theorem]{Proposition}
\newtheorem{lemma}[theorem]{Lemma}
\newtheorem{corollary}[theorem]{Corollary}

\newtheorem{example}[theorem]{Example}

\numberwithin{equation}{section}
\numberwithin{theorem}{section}

\DeclareMathOperator{\D}{D}
\DeclareMathOperator{\res}{res}

\begin{document}

\title{Lagrangian multiforms for Kadomtsev-Petviashvili (KP) and the Gelfand-Dickey hierarchy}

\author{Duncan Sleigh, Frank Nijhoff and Vincent Caudrelier\\ 
School of Mathematics, University of Leeds}

\maketitle

\abstract{We present, for the first time, a Lagrangian multiform for the complete Kadomtsev-Petviashvili (KP) hierarchy -- a single variational object that generates the whole hierarchy and encapsulates its integrability.  By performing a reduction on this Lagrangian multiform, we also obtain Lagrangian multiforms for the Gelfand-Dickey hierarchy of hierarchies, comprising, amongst others, the Korteweg-de Vries and Boussinesq hierarchies.}

\section{Introduction}

A feature of integrable systems is the existence of hierarchies of mutually compatible equations.  A significant limitation of using traditional Lagrangians for such hierarchies is that they do not capture this compatibility. This limitation was overcome by the Lagrangian multiform \cite{Lobb2009}, which allows compatible Lagrangians (i.e., Lagrangians of compatible equations) to be combined into a single variational object.  In recent years, numerous examples of Lagrangian multiforms for continuous one and two dimensional integrable hierarchies have been found (e.g. Calogero-Moser \cite{YooKong2011}, Toda \cite{Suris2017}, potential KdV \cite{Suris2016} and AKNS \cite{Sleigh2019,Vermeeren2020,ThisPaper,Caud2020}).  It is natural to expect that there should exist a Lagrangian multiform for the most well known three dimensional integrable hierarchy, the Kadomtsev-Petviashvili (KP) hierarchy \cite{KP1970,Sato1981}.  A Lagrangian multiform for the discrete KP hierarchy (the first example of a Lagrangian 3-form) was given in \cite{Quispel2009}, whilst a Lagrangian multiform for the first two flows of the continuous KP hierarchy was presented in \cite{ThisPaper}.  This continuous KP Lagrangian multiform was limited in the sense that extending it to contain higher flows of the hierarchy would result in non-local terms in the multiform, and also there was no algorithmic method to perform such an extension.\\

In \cite{Dickey1987}, Dickey gives a family of Lagrangians in terms of pseudodifferential operators for the individual equations of the KP hierarchy.  In this paper we assemble Dickey's KP Lagrangians, along with a new set of Lagrangians to create Lagrangian multiform for the full KP hierarchy.  This is the first ever example of a continuous Lagrangian 3-form for a complete integrable hierarchy.  Then, based on the reduction of KP to the Gelfand-Dickey hierarchy, we perform a reduction on the KP Lagrangian multiform to obtain Lagrangian multiforms for each of the integrable hierarchies that comprise the Gelfand-Dickey hierarchy. \\

We begin by giving a brief introduction to Lagrangian multiforms in Section \ref{LMintro} and then summarise key results relating to pseudodifferential operators in Section \ref{PDOintro}.  In Section \ref{KPGD} we introduce the KP hierarchy in terms of pseudodifferential operators, and also its reduction to the Gelfand-Dickey hierarchy.  In Section \ref{KPLagSect} we introduce Dickey's KP Lagrangian.  Our main result, a Lagrangian multiform for the KP hierarchy is given in Section \ref{KPMFSect}, followed by its reduction to Gelfand-Dickey in Section \ref{GDmfred}.

\subsection{Lagrangian multiforms}\label{LMintro}

Lagrangian multiforms were first conceived of in \cite{Lobb2009} to allow a variational description of compatible systems of equations, and have subsequently generated considerable research interest. The traditional variational approach involves a Lagrangian that is a volume form, i.e.,

\begin{equation}
\mathscr{L}(x,u^{(n)})\textsf{d}x_1\wedge\ldots\wedge \textsf{d}x_k,
\end{equation}
on a $k$-dimensional base manifold.  We use the notation $u^{(n)}$ to represent $u$ and its derivatives with respect to the independent variables $x_i$, up to the $n^{th}$ order.  This can only give as many equations of motion as there are components of $u$.  A \textit{Lagrangian multiform} 

\begin{equation}\label{form}
\textsf{M}=\sum_{1\leq i_1<\ldots <i_k\leq N}\mathscr{L}_{(i_1\ldots i_k)}(x,u^{(n)})\ \textsf{d}x_{i_1}\wedge\ldots\wedge \textsf{d}x_{i_k}.
\end{equation}
is a $k$-form in an $N$ dimensional base manifold with $k<N$, subject to the following variational principle.  We require that any $u$ that is a critical point of the action

\begin{equation}
S[u;\sigma]=\int_{\sigma}\textsf{M}(x,u^{(n)})
\end{equation}
must be a critical point for all $k$-dimensional surfaces of integration $\sigma$.  This results in the \textit{multiform Euler-Lagrange equations}, given by $\delta \textsf{dM}=0$.  Furthermore, we require that, on the equations given by $\delta \textsf{dM}=0$, any interior deformation of the surface $\sigma$ must leave the critical action $S$ unchanged (i.e., that on the equations defined by $\delta \textsf{dM}=0$, we require that $ \textsf{dM}=0$).  
The multiform Euler-Lagrange equations can also be presented as a set of equations in terms of variational derivatives of the $\mathscr{L}_{(i_1\ldots i_k)}$ that includes the usual Euler-Lagrange equations of each $\mathscr{L}_{(i_1\ldots i_k)}$.  In \cite{Vermeeren2018} and \cite{ThisPaper}, proofs are given that show the equivalence of these two presentations of the multiform Euler-Lagrange equations.  In Appendix \ref{MFEL} we go further and show explicitly the link between these two presentations of the multiform Euler-Lagrange equations.

We shall use the convention that Lagrangians $\mathscr{L}_{(i\ldots j)}$ are anti-symmetric when permuting the sub-indices so, for example, $\mathscr{L}_{(123)}=\mathscr{L}_{(312)}=-\mathscr{L}_{(132)}$.

\subsection{Pseudodifferential operators}\label{PDOintro}

The main results in this paper require the use of pseudodifferential operators.  Here we give a brief summary based on \cite[Chapter~1]{DickeyBook} and the references therein.  We introduce the differential algebra $\mathcal{A}$ with generators $u_1,u_2,u_3,\ldots$ and derivation  $\D_x$, the total derivative with respect to x, such that $\D_x u_\alpha^{(i)} =(u_\alpha^{(i)})_x=u_\alpha^{(i+1)}$, where $u_\alpha^{(0)}=u_\alpha$.  Also, $\D_x$ obeys the Leibnitz rule $\D_xu_\alpha^{(i)}u_\beta^{(j)}=u_\alpha^{(i+1)}u_\beta^{(j)}+u_\alpha^{(i)}u_\beta^{(j+1)}$.  Elements of $\mathcal{A}$ are polynomials with real or complex coefficients in the generators $u_\alpha$ and their derivatives of arbitrary order.  The operator $\partial$ is defined such that for $f\in\mathcal{A}$,

\begin{equation}\label{commrule}
\partial^kf=f\partial^k+{k\choose 1}f'\partial^{k-1}+{k\choose 2}f''\partial^{k-2}+\ldots
\end{equation}
where $f\in\mathcal{A}$, $f'=\D_x f$ and

\begin{equation}
{k\choose i}=\frac{k(k-1)\ldots(k-i+1)}{i!}.
\end{equation}
When $k>0$ this sum naturally truncates, whereas when $k<0$ the sum is infinite.  Using these definitions for $\D_x$ and $\partial$, we note that for $f\in\mathcal{A}$, $\D_x f$ is also in $\mathcal{A}$, whereas $\partial f$ is not, since $\partial f=\D_x f+f\partial$ which is an operator.\\

The ring of pseudodifferential operators $\mathcal{R}$ consists of elements
\begin{equation}
X=\sum_{i=-\infty}^mX_i\partial^i,\quad X_i \in \mathcal{A}.
\end{equation}
Elements of $\mathcal{R}$ can be added (in the natural way) and multiplied term by term, moving all $\partial$s to the right hand side according to the commutation rule given in \eqref{commrule}.  Using the commutation rule \eqref{commrule}, elements of $\mathcal{R}$ can also be written in the equivalent ``left'' form

\begin{equation}\label{leftformdef}
X=\sum_{i=-\infty}^m\partial^i\tilde X_i,\quad \tilde X_i \in \mathcal{A}.
\end{equation}

If the leading coefficient of $X$, $X_m$, is 1, then there exists a unique inverse $X^{-1}$ also with leading coefficient 1, such that $XX^{-1}=X^{-1}X=1$.  There also exists a unique $m^{th}$ root of $X$, $X^{1/m}$ starting with $\partial$.  Then $X^{p/m}=(X^{1/m})^p$ and $(X^{1/m})^m=X$.  We define $\mathcal{R}_+$ to be the set of all elements
\begin{equation}
X_+=\sum_{i=0}^mX_i\partial^i
\end{equation}
and $\mathcal{R}_-$ to be the set of all elements
\begin{equation}
 X_-=\sum_{i=-\infty}^{-1}X_i\partial^i
\end{equation}
The residue of a pseudodifferential operator, $\res \{X\}=X_{-1}$, the coefficient of $\partial^{-1}$ in $X$.  We shall make use of two important properties relating to residues.  Firstly,

\begin{equation}\label{plusminusprop}
\res \big\{X_+Y\big\}=\res \big\{X_+Y_-\big\}=\res \big\{XY_-\big\}.
\end{equation}
The second property we shall use is given on the following lemma.

\begin{lemma}\label{rescomdiv}
The residue of a commutator of two pseudodifferential operators $X$ and $Y$,

\begin{equation}\label{rescomderiv}
\res\{ [X,Y]\}=\D_x h
\end{equation}
for some $h\in \mathcal{A}$, so is a total $x$ derivative.
\end{lemma}
This lemma is given in \cite[Chapter~1]{DickeyBook} but the proof contains errors that are corrected here.

\begin{proof}
We verify this for single term pseudodifferential operators $S=s\,\partial^m$ and $T=t\,\partial^n$.  We shall use the notation $s^{(k)}=\D_x^k s$ and similarly for $t$.  We first note that $\res\{[S,T]\}$ is only non-zero if one of $m$ and $n$ is greater than or equal to zero whilst the other is negative.  Without loss of generality, we shall assume that $m\geq 0$ and $n<0$.  The product

\begin{equation}\label{XYprod}
ST=\sum_{k=0}^\infty {m\choose k}st^{(k)}\partial^{m+n-k}.
\end{equation}
so
\begin{equation}
\res\{ST\}={m\choose m+n+1}st^{(m+n+1)}
\end{equation}
when $m+n+1\geq 0$.  Otherwise $\res\{ST\}=0$ since $k\geq 0$ in \eqref{XYprod}.  It follows that
\begin{equation}
\res\{[S,T]\}={m\choose m+n+1}st^{(m+n+1)}-{n\choose m+n+1}st^{(m+n+1)}.
\end{equation}
We notice that

\begin{equation}
{m\choose m+n+1}=\frac{m(m-1)\ldots(-n)}{(m+n+1)!}\quad \text{and}\quad{n\choose m+n+1}=\frac{n(n-1)\ldots(-m)}{(m+n+1)!}
\end{equation}
so

\begin{equation}
{n\choose m+n+1}=(-1)^{m+n+1}{m\choose m+n+1}.
\end{equation}
Then
\begin{equation}
\begin{split}
\res\{[S,T]\}=&{m\choose m+n+1}(st^{(m+n+1)}+(-1)^{m+n}st^{(m+n+1)})\\
=&{m\choose m+n+1}(st^{(m+n+1)}+s^{(1)}t^{(m+n)}-s^{(1)}t^{(m+n)}-s^{(2)}t^{(m+n-1)}+s^{(2)}t^{(m+n-1)}+\ldots\\&
\ldots-(-1)^{m+n}t^{(1)}s^{(m+n)}+(-1)^{m+n}t^{(1)}s^{(m+n)}+(-1)^{m+n}ts^{(m+n+1)})
\end{split}
\end{equation}
where, to get the expression on the second line we have added and subtracted $\sum_{\alpha=1}^{m+n} s^{(\alpha)}t^{(m+n+1-\alpha)}$.  We  recognise this as a total $x$ derivative, so
\begin{equation}\label{hproc}
\res\{[S,T]\}={m\choose m+n+1}\D_x\sum_{\alpha=0}^{m+n}(-1)^\alpha s^{(\alpha)}t^{(m+n-\alpha)}.
\end{equation}
It follows that, for general pseudodifferential operators $X$ and $Y$, their residue, $\res\{[X,Y]\}$ can be expressed as the sum of total derivatives of the form given in \eqref{hproc} for pairs $X_i$ and $Y_j$, so is a total $x$ derivative.
\end{proof}

\section{The KP hierarchy and its reduction to Gelfand-Dickey}\label{KPGD}

\subsection{The KP hierarchy}

Here we give a brief summary of Sato's scheme \cite{Sato1981} for the KP hierarchy \cite{KP1970}.  We let

\begin{equation}\label{KPL}
L=\partial+u_1\partial^{-1}+u_2\partial^{-2}+\ldots=\partial+\sum_{\alpha=1}^\infty u_\alpha\partial^{-\alpha}.
\end{equation}
Using the notation $L^i_+$ to represent $(L^i)_+$, for $i>0$

\begin{equation}\label{KPsingle}
L_{x_i}=[L^i_+,L]
\end{equation}
gives us the KP hierarchy.  For each $i$, this produces an infinite set of PDEs containing derivatives with respect to $x_i$ and $x$.  From the case where $i=1$, we see that $L_{x_1}=\D_x L$, allowing us to identify $x_1$ with $x$.  A consequence of \eqref{KPsingle} is that

\begin{equation}\label{Ltothen}
(L^n)_{x_i}=[L^i_+,L^n]
\end{equation}
for all $n\geq1$.  This can be proved by induction on $n$.  It follows that

\begin{equation}
\begin{split}
(L_+^j)_{x_i}-(L_+^i)_{x_j}&=[L^i_+,L^j]_+-[L^j_+,L^i]_+\\
&=[L^i_+-L^i,L^j]_++[L^i,L^j_+]_+\\
&=[-L^i_-,L^j]_++[L^i,L^j_+]_+\\
&=[-L^i_-,L^j_+]_++[L^i,L^j_+]_+\\
&=[L^i_+,L^j_+].
\end{split}
\end{equation}
This gives us the ``zero-curvature'' equations for KP,

\begin{equation}\label{KPdouble}
(L^j_+)_{x_i}-(L^i_+)_{x_j}=[L^i_+,L^j_+].
\end{equation}
For each $i,j>0$, this produces a finite set of PDEs containing derivatives with respect to $x_i$, $x_j$ and $x$.  In the case where $i=2$ and $j=3$, \eqref{KPdouble} gives us

\begin{equation}\label{KPsimult}
\begin{split}
&3(u_1)_{x_2}=3u_1^{(2)}+6u_2^{(1)}\\
&3(u_1^{(1)})_{x_2}+3(u_2)_{x_2}-2(u_1)_{x_3}=u_1^{(3)}+3u_2^{(2)}-6u_1u_1^{(1)}.
\end{split}
\end{equation}
Letting $2u_1=u$ and eliminating $u_2$, this gives us

\begin{equation}\label{KPequation}
3u_{x_2x_2}=(4u_{x_3}-u^{(3)}-6uu^{(1)})_x,
\end{equation}
the KP equation that gives its name to the hierarchy.\\

For a fixed choice of $i$ and $j$, the PDEs given by \eqref{KPsingle} for $i$ and $j$ are not equivalent to the PDEs given by \eqref{KPdouble} for the same $i$ and $j$, since \eqref{KPsingle} gives an infinite set of PDEs whilst \eqref{KPdouble} gives a finite one. However the set of PDEs given by \eqref{KPsingle} for all $i>0$ is equivalent to the set of PDEs given by \eqref{KPdouble} for all $i,j>0$.  We have already shown that we can obtain \eqref{KPdouble} from \eqref{KPsingle}.  The following lemma relates to the converse.

\begin{lemma}\label{lemmazerotolax}
The set of equations given by

\begin{equation}\label{lemmaKP2}
(L^j_+)_{x_i}-(L^i_+)_{x_j}=[L^i_+,L^j_+].
\end{equation}
for all $1\leq i<j$ is equivalent to the set of equations given by

\begin{equation}\label{lemmaKP1}
L_{x_i}=[L^i_+,L]
\end{equation}
for all $i\geq 1$.

\end{lemma} 

\begin{proof}
We have already shown that \eqref{lemmaKP1} for $i$ and $j$ implies \eqref{lemmaKP2} for the same $i$ and $j$.  To show that \eqref{lemmaKP2} for all $1\leq i,j$ implies \eqref{lemmaKP1} for all $i\geq 1$, we consider \eqref{KPdouble} in the form

\begin{equation}
(L_+^j)_{x_i}-(L_+^i)_{x_j}=[L^i_+,L^j]_+-[L^j_+,L^i]_+,
\end{equation}
and without loss of generality assume that $j>i$.  The first $j-i$ terms of this (i.e. the coefficients of $\partial^{k}$ for $k$ from $i-1$ to $j-2$) are identical to the first $j-i$ terms of

\begin{equation}\label{Firstterms}
L^j_{x_i}=[L^i_+,L^j].
\end{equation}
We now let $j=n+1$ in \eqref{Firstterms} and multiply from the left by $L^{-n}$, and from this we subtract \eqref{Firstterms} with $j=n$, multiplied on the left by $L^{-n}$, and on the right by $L$ to obtain

\begin{equation}
\begin{split}
L^{-n}(L^{n+1}_{x_i}-L^n_{x_i}L)=L^{-n}([L^i_+,L^{n+1}]-[L^i_+,L^n]L).
\end{split}
\end{equation}
The left hand side of this is just $L_{x_i}$, whilst the right hand side simplifies to $[L^i_+,L]$.  Therefore two copies of \eqref{KPdouble} with $j=n$ and $j=n+1$ gives us the first $n-i$ terms of

\begin{equation}
L_{x_i}=[L^i_+,L].
\end{equation}
Since $n$ is arbitrary, we are able to obtain all terms of \eqref{KPsingle}.

\end{proof}

In \cite{ThisPaper}, a Lagrangian multiform incorporating a re-scaled version of \eqref{KPequation} and the corresponding equation arising from \eqref{KPdouble} with $i=2$ and $j=4$ was presented  with the following Lagrangian coefficients:

\begin{subequations}
\begin{equation}
\mathscr{L}_{(123)}=\frac{1}{2}v_{x_1x_1}v_{x_1x_3}-\frac{1}{2}v_{3x_1}^2-\frac{1}{2}v_{x_1x_2}^2+v_{x_1x_1}^3
\end{equation}

\begin{equation}
\mathscr{L}_{(412)}=\frac{1}{2}v_{x_1x_1}v_{x_1x_4}-2v_{3x_1}v_{x_1x_1x_2}-\frac{2}{3}v_{x_1x_2}v_{x_2x_2}+4v_{x_1x_1}^2v_{x_1x_2}
\end{equation}

\begin{equation}
\begin{split}
\mathscr{L}_{(234)}=&-\frac{1}{2}v_{x_1x_3}v_{x_1x_4}-4v_{x_1x_3}v_{3x_1x_2}+2v_{x_1x_1x_3}v_{x_1x_1x_2}-\frac{2}{3}v_{x_2x_2}v_{x_2x_3}+v_{x_2x_2}v_{x_1x_4}\\
&+4v_{x_2x_2}v_{3x_1x_2}-\frac{8}{3}v_{x_1x_2x_2}v_{x_1x_1x_2}-v_{3x_1}v_{x_1x_1x_4}+\frac{4}{3}v_{3x_1}v_{3x_2}-4v_{3x_1}^2v_{x_1x_2}\\
&+8v_{x_1x_1}v_{3x_1}v_{x_1x_1x_2}+8v_{x_1x_1}v_{x_1x_2}v_{x_2x_2}+\frac{4}{3}v_{x_1x_2}^3-8v_{x_1x_1}v_{x_1x_2}v_{x_1x_3}-8v_{x_1x_1}^3v_{x_1x_2}
\end{split}
\end{equation}

\begin{equation}
\begin{split}
\mathscr{L}_{(341)}=&\frac{2}{3}v_{x_2x_2}^2+2v_{4x_1}^2-2v_{3x_1}v_{x_1x_1x_3}-\frac{4}{3}v_{x_2x_2}v_{x_1x_3}-\frac{2}{3}v_{x_1x_2}v_{x_2x_3}+v_{x_1x_2}v_{x_1x_4}\\
&-\frac{4}{3}v_{x_1x_1x_2}^2+\frac{4}{3}v_{3x_1}v_{x_1x_2x_2}+12v_{x_1x_1}^2v_{4x_1}+4v_{3x_1}^2v_{x_1x_1}-4v_{x_1x_1}^2v_{x_2x_2}\\
&+4v_{x_1x_1}v_{x_1x_2}^2+4v_{x_1x_1}^2v_{x_1x_3}+10v_{x_1x_1}^4.
\end{split}
\end{equation}
\end{subequations}
where the dependent variable $v_{x_1x_1} =u$ has been used to eliminate non-local terms.  These Lagrangians were found using the variational symmetries method outlined in the same paper.  Although it is possible to extend this Lagrangian multiform to incorporate more flows of the hierarchy, the resultant Lagrangians become increasingly unwieldy.  Also, as we progress up the hierarchy, an ever increasing number of non-local terms appear in the Lagrangians, and the Lagrangians grow very large very quickly.  Expanding this multiform to include the $x_5$ flow results in Lagrangians that are many pages long.  Also, this approach does not yield an explicit formula for all of the constituent Lagrangians of the multiform for the complete hierarchy, so in order to obtain a multiform for the entire hierarchy, a different approach is needed.

\subsection{The Gelfand-Dickey hierarchy as a reduction of KP}

The $n^{th}$ Gelfand-Dickey hierarchy \cite{Gelfand1976} can be formulated as follows.  We let

\begin{equation}
L_{GD}=\partial^n+v_{n-2}\partial^{n-2}+v_{n-3}\partial^{n-3}+\ldots+v_0
\end{equation}
and let

\begin{equation}
P_m=(L_{GD}^{m/n})_+.
\end{equation}
We note that whilst $L_{GD}$ is not a pseudodifferential operator, in general a fractional power of $L_{GD}$ will be.  The $n^{th}$ Gelfand-Dickey hierarchy is then given by

\begin{equation}\label{GDeqn}
(L_{GD})_{x_m}=[P_m,L_{GD}].
\end{equation}
In the case where $n=2$, this gives the KdV hierarchy, whilst for $n=3$ we get the Boussinesq hierarchy.  We now consider the KP equation \eqref{Ltothen}

\begin{equation}\label{KPforGD}
L^n_{x_m}=[L^m_+,L^n].
\end{equation}
In order to reduce the KP hierarchy to the $n^{th}$ Gelfand-Dickey hierarchy we impose the constraint that $L^n_-=0$.  We note that 

\begin{equation}
L^n_-=0\implies L^n=L^n_+,
\end{equation}
 an $n^{th}$ order differential operator that we equate with $L_{GD}$.  It follows that $L_{GD}^{1/n}=L$, so $P_m$ is given by $L^m_+$, making \eqref{GDeqn} and the right hand expression in \eqref{KPforGD} equivalent.  We also note that $L^n_-=0\implies L^{kn}_-=0$ for all $k \in \mathbb{Z}_+$, so \eqref{KPforGD} gives $L^n_{x_m}=0$ whenever $n$ divides $m$. This is as expected since, by \eqref{GDeqn}, $(L_{GD})_{x_m}=0$ whenever $P_m$ is an integer power of $L_{GD}$, which happens when $n$ divides $m$.

\section{A Lagrangian for the KP hierarchy}\label{KPLagSect}

In this section, we present a Lagrangian for the KP hierarchy that was originally given in \cite{Dickey1987}.  We define $\mathcal{A}_\varphi$ to be the differential algebra analogous to $\mathcal{A}$ with generators $\varphi_0,\varphi_1,\varphi_2,\ldots$ (i.e. where elements of $\mathcal{A}_\varphi$ are differential polynomials in the generators $\varphi_\beta$), and we define $\mathcal{R}_\varphi$ to be the ring of pseudodifferential operators with coefficients in $\mathcal{A}_\varphi$.  We define $\mathcal{R}_{\varphi+}$ and $\mathcal{R}_{\varphi-}$ analogously to $\mathcal{R}_{+}$ and $\mathcal{R}_{-}$.  We make the dressing substitution

\begin{equation}\label{Lphidef}
L=\phi\partial\phi^{-1}
\end{equation}
where

\begin{equation}\label{phidef}
\phi=1+\sum_{\beta=0}^\infty\varphi_\beta\partial^{-\beta-1},
\end{equation}
noting that because of the leading 1, a unique $\phi^{-1}$ exists.  Expanding \eqref{Lphidef} we find that

\begin{equation}\label{varphiexp}
L=\partial -\varphi_0'\partial^{-1}+(\varphi_0\varphi_0'-\varphi_1')\partial^{-2}+(\varphi_1\varphi_0'+\varphi_0\varphi_1'-(\varphi_0')^2-\varphi_0^2\varphi_0'-\varphi_2')\partial^{-3}+\ldots,
\end{equation}
where $\varphi_\beta'$ denotes the $x$ derivative of $\varphi_\beta$.  Equating coefficients with \eqref{KPL}, we see that  $u_1=-\varphi_0'$, $u_2=\varphi_0\varphi_0'-\varphi_1'$, $u_3=\varphi_1\varphi_0'+\varphi_0\varphi_1'-(\varphi_0')^2-\varphi_0^2\varphi_0'-\varphi_2'$ etc., giving an injective map from $\mathcal{A}$ to $\mathcal{A}_\varphi$.\\  

In order to determine the resulting KP equation in terms of $\phi$, we invoke the idea of homogeneity in the sense of all terms of an expression carying equal weight.  Let us consider this in the case of the KP equation
\begin{equation}
3u_{x_2x_2}=(4u_{x_3}-u^{(3)}-6uu^{(1)})_x.
\end{equation}
We begin by assigning a weight of 1 to the derivative with respect to $x$.  On the left hand side of the equation, we see a $u_{x_2x_2}$ term, which we compare to the $u^{(4)}$ term on the right hand side.  In order for these terms to have equal weight, an $x_2$ derivative must have weight 2.  Similarly, by comparing the $u_{x_3}^{(1)}$ and $u^{(4)}$ terms, it follows that an $x_3$ derivative has weight 3.  Finally by comparing $u^{(3)}$ and $uu^{(1)}$ we see that $u$ carries weight 2.  Whenever it is possible to assign weights in this manner such that all terms of an expression carry equal weight, we say that the expression is \textit{homogeneous}. \\

Homogeneity can also be introduced directly on the level of the pseudodifferential operators.  Applying this to the KP operator

\begin{equation}
L=\partial+u_1\partial^{-1}+u_2\partial^{-2}+\ldots,
\end{equation}
we again assign a weight of 1 to the derivative with respect to $x$, so the leading $\partial$ carries weight 1.  In order for all terms to carry equal weight, it follows that $u_1$ has weight 2, $u_2$ has weight 3, and in general $u_\alpha$ has weight $\alpha+1$.  Similarly, the leading 1 of the operator

\begin{equation}
\phi=1+\varphi_0\partial^{-1}+\varphi_1\partial^{-2}+\ldots
\end{equation}
tells us that $\phi$ has weight 0, so $\varphi_0$ has weight 1, $\varphi_1$ has weight 2, and $\varphi_\beta$ has weight $\beta+1$ in order that each term has weight 0.  In this paper we only deal with homogeneous equations.  With this in mind, we have the following lemma.

\begin{lemma}\label{laxtophi}
We let $L=\phi\partial\phi^{-1}\in \mathcal{R}_\varphi$.  Then

\begin{equation}
L_{x_i}=[L^i_+,L]  \iff  \phi_{x_i}=-L^i_-\phi .  
\end{equation}

\end{lemma}

\begin{proof}
Using that $L=\phi\partial\phi^{-1}$, the equation
\begin{equation}\label{weightKP}
L_{x_i}=[L^i_+,L]
\end{equation}
becomes

\begin{equation}
[\phi_{x_i}\phi^{-1}-L_+^i,L]=0,
\end{equation}
This is equivalent to the statement that 

\begin{equation}\label{weightphi}
\phi_{x_i}\phi^{-1}-L_+^i+f_i=0
\end{equation}
for some $f_i$ in $\mathcal{R}_\varphi$ such that $[L,f_i]=0$.  Letting $\tilde f_i=\phi^{-1}f_i\phi$, the requirement that $[L,f_i]=0$ is equivalent to the requirement that $[\partial,\tilde f_i]=\D_x\tilde f_i=0$.  Therefore $\tilde f_i$ is a constant in $\mathcal{R}_\varphi$, so 

\begin{equation}
\tilde f_i=\sum_{j=-\infty}^m\gamma_j\partial^j
\end{equation}
for some $m$, where each $\gamma_j$ is a constant in $\mathcal{A}_\varphi$ (i.e. a real or complex number), and consequently

\begin{equation}
f_i=\sum_{j=-\infty}^m\gamma_jL^j
\end{equation}
for the same constants $\gamma_j$.  In \eqref{weightphi} we see that both $\phi_{x_i}\phi^{-1}$ and $L_+^i$ are of weight $i$, so we require that $f_i$ is also of weight $i$.  Therefore, $\gamma_j=0$ whenever $j\neq i$, so $f_i$ is of the form $\gamma_i L^i$.  When $f_i$ takes this form, the coefficient of $\partial^i$ in \eqref{weightphi} is $\gamma_i-1$, and setting this equal to zero gives us that $\gamma_i=1$.  Then \eqref{weightphi} becomes

\begin{equation}
\phi_{x_i}\phi^{-1}+L^i_-=0,
\end{equation}
so the resulting equation for $\phi_{x_i}$ is

\begin{equation}\label{phiequation}
\phi_{x_i}=-L^i_-\phi.
\end{equation}
I.e., 
\begin{equation}
L_{x_i}=[L^i_+,L]  \implies \phi_{x_i}=-L^i_-\phi.  
\end{equation}
Conversely, we see that if \eqref{phiequation} holds then
\begin{equation}
\begin{split}
L_{x_i}&=(\phi\partial\phi^{-1})_{x_i}\\
&=\phi_{x_i}\partial\phi^{-1}-\phi\partial\phi^{-1}\phi_{x_i}\phi^{-1}\\
&=-\L^i_-\phi\partial\phi^{-1}+\phi\partial\phi^{-1}\L^i_-\\
&=[-L^i_-,L]\\
&=[L^i_+,L]
\end{split}
\end{equation}
so \eqref{phiequation} implies \eqref{weightKP}.
\end{proof}

\begin{corollary}\label{KPeqcoroll}
Lemmas  \ref{lemmazerotolax} and \ref{laxtophi} together tell us that the set of equations given by

\begin{equation}
(L^j_+)_{x_i}-(L^i_+)_{x_j}=[L^i_+,L^j_+]  
\end{equation}
in $\mathcal{R}$ for all $1\leq i,j$ is equivalent to the set of equations given by

\begin{equation}
\phi_{x_i}\phi^{-1}+L^i_-=0
\end{equation}
in $\mathcal{R}_\varphi$ for all $i\geq 1$.

\end{corollary}

We now consider a Lagrangian $\mathscr{L}_{(1ij)}\textsf{d}x_1\wedge \textsf{d}x_i\wedge \textsf{d}x_j$ with $\mathscr{L}_{(1ij)}\in \mathcal{A}_\varphi$.  For such a Lagrangian, we can take variational derivatives $\dfrac{\delta \mathscr{L}_{(1ij)}}{\delta \varphi_\beta}$ (i.e., the Euler operator with respect to $\varphi_\beta$ acting on $\mathscr{L}_{(1ij)}$) to obtain expressions in $\mathcal{A}_\varphi$.  However it is convenient to define the variational derivative with respect to the pseudodifferential operator $\phi$,

\begin{equation}\label{varderiv}
\frac{\delta \mathscr{L}_{(1ij)}}{\delta \phi}=\sum_{\beta=0}^\infty\partial^\beta\frac{\delta \mathscr{L}_{(1ij)}}{\delta \varphi_\beta}.
\end{equation}
According to this definition, $\dfrac{\delta \mathscr{L}_{(1ij)}}{\delta \phi}$ is a pseudodifferential operator in $\mathcal{R}_{\varphi+}$ that can be put in the usual form with all $\partial$s on the right using \eqref{commrule}.  The motivation for this definition is made clear by the following lemma.

\begin{lemma}
If there exist $h_1$, $h_2$ and $h_3$ such that

\begin{equation}\label{deltaLdef}
\delta \mathscr{L}_{(1ij)} = \res\{X\,\delta\phi\}+\D_xh_1+\D_{x_i}h_2+\D_{x_j}h_3
\end{equation}
for some $X\in\mathcal{R}_\varphi$, then the variational derivative of $\mathscr{L}_{(1ij)}$ with respect to $\phi$,

\begin{equation}
\frac{\delta \mathscr{L}_{(1ij)}}{\delta \phi}=X_+
\end{equation}

\end{lemma}

\begin{proof}
Since $\delta\phi=\delta\varphi_0\partial^{-1}+\delta\varphi_1\partial^{-2}+\ldots$ has only negative powers of $\partial$, \eqref{deltaLdef} is equivalent to

\begin{equation}
\delta \mathscr{L}_{(1ij)}=\res\{X_+\,\delta\phi\}+\D_xh_1+\D_{x_i}h_2+\D_{x_j}h_3.
\end{equation}
We write $X_+$ in the ``left'' form described in equation \eqref{leftformdef}, so

\begin{equation}
X_+=\sum_{k=0}^m\partial^k\tilde X_k,\quad \tilde X_k \in \mathcal{A}_\varphi,
\end{equation}
and consider the product of an arbitrary term in $X_+$ with an arbitrary term in $\delta \phi$.  This will be of the form

\begin{equation}
\partial^n\tilde X_n \,\delta\varphi_m\partial^{-m-1}=\tilde X_n \,\delta\varphi_m\partial^{n-m-1}+\sum_{i=1}^n{n\choose i}\D_x^i(\tilde X_n \,\delta\varphi_m)\partial^{n-m-i-1}
\end{equation}
and the only term on the right hand side that is not a total derivative is $\tilde X_n \,\delta\varphi_m\partial^{n-m-1}$.  Therefore,

\begin{equation}
\delta \mathscr{L}_{(1ij)}=\res\{X_+\,\delta\phi\}+\D_xh_1+\D_{x_i}h_2+\D_{x_j}h_3=\sum_{k=0}^m\tilde X_k \,\delta\varphi_k+\D_x\tilde h_1+\D_{x_i}h_2+\D_{x_j}h_3
\end{equation}
for some $\tilde h_1$, so the variational derivative

\begin{equation}
\frac{\delta\mathscr{L}_{(1ij)}}{\delta \varphi_k}=\tilde X_k
\end{equation}
for $0\leq k\leq m$ and is zero for $k>m$.  It follows that 

\begin{equation}
\frac{\delta \mathscr{L}_{(1ij)}}{\delta \phi}=\sum_{k=0}^\infty\partial^k\frac{\delta\mathscr{L}_{(1ij)}}{\delta \varphi_k}=\sum_{k=0}^m\partial^k\tilde X_k=X_+
\end{equation}

\end{proof}
Following the formulation in \cite{Dickey1987}, we introduce

\begin{equation}
\phi_p=1+p\sum_{\beta=0}^\infty\varphi_\beta\partial^{-\beta-1}.
\end{equation}
where $p\in \mathbb{R}$.

\begin{proposition}\label{DickekLagProp}
The Lagrangian density

\begin{equation}\label{KPLag}
\mathscr{L}_{(1ij)}=\res\bigg\{-\int_0^1p^{-1}[(\phi_p\partial^i\phi_p^{-1})_+,(\phi_p\partial^j\phi_p^{-1})_+]\phi_p^{-1}\textsf{d} p+\partial^j\phi^{-1}\phi_{x_i}-\partial^i\phi^{-1}\phi_{x_j}\bigg\}
\end{equation}
gives Euler-Lagrange equations that are equivalent to the KP equation 

\begin{equation}
(L^i_+)_{x_j}-(L_+^j)_{x_i}+[L^i_+,L^j_+]=0.
\end{equation}

\end{proposition}

It is important to note that where $\partial$ appears in this Lagrangian, it signifies an operator that acts on everything to its right, rather than the $x$ derivative of whatever is immediately to its right.  Also, even though $\phi$ consists of an infinite number of components, because this Lagrangian is a residue, only a finite number of these components actually feature.  A proof that \eqref{KPLag} gives the KP equation as its Euler-Lagrange equations is given in \cite{Dickey1987} and repeated here.  We shall require the following lemma:
\begin{lemma}\label{varprop}
The following formula holds:
\begin{equation}\label{int}
\delta\res\bigg\{\int_0^p\tilde p^{-1}[(\phi_{\tilde p}\partial^i\phi_{\tilde p}^{-1})_+,(\phi_{\tilde p}\partial^j\phi_{\tilde p}^{-1})_+]\phi_{\tilde p}^{-1}\textsf{d} \tilde p\bigg\}=-\res\big\{[(\phi_p\partial^i\phi_p^{-1})_+,(\phi_p\partial^j\phi_p^{-1})_+]\delta\phi_p \ \phi_p^{-1}\big\}+\D_x h_1
\end{equation}
with 
\begin{equation}\label{h1formula}
\begin{split}
h_1=\ &  \int\int_0^p\tilde p^{-1}\res\big\{[T[V,S],U]+[[T,U]_+S,V]+[U[V,S]_+,T]+[UT,[V,S]_+]+[T[S,U],V]\\&+[U,[T,V]_+S]+[V[S,U]_+,T]+[VT,[S,U]_+]+[[U,V],TS]+[T,[U,V]S]\big\}\textsf{d} \tilde p\,\textsf{d}x.
\end{split}
\end{equation}
where $S=\phi_{\tilde p}^{-1}$, $T=\delta\phi_{\tilde p} \ \phi_{\tilde p}^{-1}$, $U=(\phi_{\tilde p}\partial^i\phi_{\tilde p}^{-1})_+$ and $V=(\phi_{\tilde p}\partial^j\phi_{\tilde p}^{-1})_+$. This $h_1$ is local.
\end{lemma}
The first part of this result is essentially the same as the one given by Dickey in \cite{Dickey1987}.  However, Dickey does not give an explicit expression for $h_1$, since when considering a single Lagrangian, it is only necessary to show that it is a total $x$ derivative.  In the Lagrangian multiform case, we will require an expression for $h_1$, so it is included here. 

\begin{proof}{of Lemma \ref{varprop}.}
We proceed by taking the $p$ derivative of
\begin{equation}\label{varintprop}
\delta\res\bigg\{\int_0^p\tilde p^{-1}[(\phi_{\tilde p}\partial^i\phi_{\tilde p}^{-1})_+,(\phi_{\tilde p}\partial^j\phi_{\tilde p}^{-1})_+]\phi_{\tilde p}^{-1}\textsf{d} \tilde p\bigg\}+\res\big\{[(\phi_p\partial^i\phi_p^{-1})_+,(\phi_p\partial^j\phi_p^{-1})_+]\delta\phi_p \ \phi_p^{-1}\big\},\end{equation}
multiplying by $p$, and using that $p\dfrac{\partial\phi_p}{\partial p}=\phi_p-1$. This gives us
\begin{equation}\label{deltaint}
\begin{split}
&\delta\res\big\{[(\phi_{p}\partial^i\phi_p^{-1})_+,(\phi_p\partial^j\phi_p^{-1})_+]\phi_p^{-1}\big\}+\res\big\{[(\phi_p\partial^i\phi_p^{-1})_+,(\phi_p\partial^j\phi_p^{-1})_+]\delta\phi_p \ \phi_p^{-2}\big\}\\&+\res\big\{\big( p\frac{\partial}{\partial p}[(\phi_p\partial^i\phi_p^{-1})_+,(\phi_p\partial^j\phi_p^{-1})_+]\big)\delta\phi_p \ \phi_p^{-1}\big\}.
\end{split}
\end{equation}
Again using $p\dfrac{\partial\phi_p}{\partial p}=\phi_p-1$ we find that

\begin{equation}
p\dfrac{\partial}{\partial p}(\phi_p\partial^i\phi_p^{-1})_+=-[\phi_p^{-1},(\phi_p\partial^i\phi_p^{-1})_+]_+.
\end{equation}
We shall also use that

\begin{equation}
\delta(\phi_p\partial^i\phi_p^{-1})_+=[\delta\phi_p \ \phi_p^{-1},(\phi_p\partial^i\phi_p^{-1})_+]_+.
\end{equation}
Letting $S=\phi_p^{-1}$, $T=\delta\phi_p \ \phi_p^{-1}$, $U=(\phi_p\partial^i\phi_p^{-1})_+$ and $V=(\phi_p\partial^j\phi_p^{-1})_+$, \eqref{deltaint} is equivalent to

\begin{equation}\label{jacobi}
\res\big\{[[T,U]_+,V]S+[U,[T,V]_+]S+[U,V]TS-[U,V]ST-[[S,U]_+,V]T-[U,[S,V]_+]T\big\}
\end{equation}
In order to show that this is a total $x$ derivative, we make use of \eqref{rescomderiv}, the property that the residue of a commutator is a total $x$ derivative.  We consider \eqref{jacobi} two terms at a time.  Firstly,

\begin{equation}\label{first}
\begin{split}
&\res\{[[T,U]_+,V]S-[U,[S,V]_+]T\}\\
=&\res\{[T,U]_+[V,S]+[[T,U]_+S,V]+[T,U][V,S]_++[U[V,S]_+,T]+[UT,[V,S]_+]\}\\
=&\res\{[T,U][V,S]+[[T,U]_+S,V]+[U[V,S]_+,T]+[UT,[V,S]_+]\}\\
=&\res\{T[U,[V,S]]+[T[V,S],U]+[[T,U]_+S,V]+[U[V,S]_+,T]+[UT,[V,S]_+]\}.
\end{split}
\end{equation}
Then
\begin{equation}\label{second}
\begin{split}
&\res\{[U,[T,V]_+]S-[[S,U]_+,V]T\}\\
=&\res\{[T,V]_+[S,U]+[U,[T,V]_+S]+[T,V][S,U]_++[V[S,U]_+,T]+[VT,[S,U]_+]\}\\
=&\res\{[T,V][S,U]+[U,[T,V]_+S]+[V[S,U]_+,T]+[VT,[S,U]_+]\}\\
=&\res\{T[V,[S,U]]+[T[S,U],V]+[U,[T,V]_+S]+[V[S,U]_+,T]+[VT,[S,U]_+]\}.
\end{split}
\end{equation}
Finally,

\begin{equation}\label{third}
\begin{split}
&\res\{[U,V]TS-[U,V]ST\}\\
=&\res\{[U,V][T,S]\}\\
=&\res\{T[S,[U,V]]+[[U,V],TS]+[T,[U,V]S]\}.
\end{split}
\end{equation}
Adding \eqref{first}, \eqref{second} and \eqref{third} together, we notice that

\begin{equation}
\res\{T([U,[V,S]]+[V,[S,U]]+[S,[U,V]])\}=0
\end{equation}
by the Jacobi identity, so \eqref{jacobi} is equal to

\begin{equation}\label{totalxderiv}
\begin{split}
&\res\{[T[V,S],U]+[[T,U]_+S,V]+[U[V,S]_+,T]+[UT,[V,S]_+]+[T[S,U],V]+[U,[T,V]_+S]\\
&+[V[S,U]_+,T]+[VT,[S,U]_+]+[[U,V],TS]+[T,[U,V]S]\}.
\end{split}
\end{equation}
Since every term is the residue of a commutator, this is a total $x$ derivative.  We set $h_1$ equal to the local expression obtained by letting $p\to \tilde p$ in \eqref{totalxderiv}, integrating with respect to $\tilde p$ from $0$ to $p$, integrating with respect to $x$ and setting the constant of integration equal to zero (i.e., the expression given in \eqref{h1formula}). It follows that, for this choice of $h_1$, \eqref{int} holds.
\end{proof}

\begin{proof}{of Proposition \ref{DickekLagProp}.}  We use Lemma \ref{varprop} with $p=1$ to obtain

\begin{equation}\label{1stpart}
\delta\res\bigg\{\int_0^1p^{-1}[(\phi_p\partial^i\phi_p^{-1})_+,(\phi_p\partial^j\phi_p^{-1})_+]\phi_p^{-1}\textsf{d} p\bigg\}=-\res\big\{[(\phi\partial^i\phi^{-1})_+,(\phi\partial^j\phi^{-1})_+]\delta\phi \ \phi^{-1}\big\}+\D_x (h_1|_{p=1}).
\end{equation}
Variation of the rest of the Lagrangian \eqref{KPLag} gives us

\begin{equation}\label{2ndpart}
\begin{split}
&\delta\res\{\partial^j\phi^{-1}\phi_{x_i}-\partial^i\phi^{-1}\phi_{x_j}\}\\
=&\D_{x_i}\res\{\partial^j\phi^{-1}\delta\phi\}-\D_{ x_j}\res\{\partial^i\phi^{-1}\delta\phi\}\\
&+\res\{\phi\partial^j\phi^{-1}\phi_{x_i}\phi^{-1}\delta\phi\,\phi^{-1}\}-\res\{\phi\partial^i\phi^{-1}\phi_{x_j}\phi^{-1}\delta\phi\,\phi^{-1}\}\\
&-\res\{\phi_{x_i}\partial^j\phi^{-1}\delta\phi\,\phi^{-1}\}+\res\{\phi_{x_j}\partial^i\phi^{-1}\delta\phi\,\phi^{-1}\}+\partial h_2\\
=&\D_{ x_i}\res\{\partial^j\phi^{-1}\delta\phi\}-\D_{ x_j}\res\{\partial^i\phi^{-1}\delta\phi\}\\
&+\res\{((L^i_+)_{x_j}-(L^j_+)_{ x_i})\delta\phi\,\phi^{-1}\}+\D_x h_2,
\end{split}
\end{equation}
where we have made use of \eqref{plusminusprop} and the fact that $\delta\phi\,\phi^{-1}\in\mathcal{R}_-$ to obtain the the final expression.  Combining \eqref{1stpart} and \eqref{2ndpart} we get

\begin{equation}
\begin{split}
\delta \mathscr{L}_{(1ij)}=&\res\{\big((L^i_+)_{x_j}-(L^j_+)_{ x_i}+[L^i_+,L^j_+]\big)\delta\phi\,\phi^{-1}\}\\
=&\res\{\phi^{-1}\big((L^i_+)_{x_j}-(L^j_+)_{ x_i}+[L^i_+,L^j_+]\big)\delta\phi\}+\D_x h_3,
\end{split}
\end{equation}
so

\begin{equation}\label{standardEL}
\frac{\delta \mathscr{L}_{(1ij)}}{\delta \phi}=\{\phi^{-1}\big((L^i_+)_{x_j}-(L^j_+)_{ x_i}+[L^i_+,L^j_+]\big)\}_+,
\end{equation}
and when set equal to zero, this is equivalent to \eqref{KPdouble}.
\end{proof}

\begin{example}The explicit form of $\mathscr{L}_{(123)}$ given by \eqref{KPLag} is

\begin{equation}
\begin{split}
\mathscr{L}_{{(123)}} =&-U_{{{ xxx_3}}}+X_{{x_2}}-VU_{{{ xx_2}}}-WU_{{x_2}}-VV_{{x_2}}-{U}^{2}U_{{x_3}}+VU_{{x_3}}+UU_{{{ xx_3}}}
+{U}^{2}U_{{{ xx_2}}}\\&+UV_{{x_3}}+{U}^{2}V_{{x_2}}-UU_{{{ xxx_2}}}-{U}^{3}U_{{x_2}}-UW_{{x_2}}-2UV_{{{ xx_2}}}-3V_{{x}}U_{{x_2}}-3U_{{{ xx}}}U_{{x_2}}
+2U_{{x}}U_{{x_3}}\\&-3U_{{x}}V_{{x_2}}-3U_{{x}}U_{{{ xx_2}}}-W_{{x_3}}+U_{{{ xxxx_2}}}-\frac{3}{2}UV_{{{ xxx}}}
-\frac{3}{2}U_{{{ xxx}}}V-3V_{{{ xx}}}V-\frac{3}{2}{U_{{x}}}^{2}{U}^{2}\\&+2U_{{{ xxx}}}{U}^{2}+2V_{{{ xx}}}{U}^{2}+2{U_{{x}}}^{2}V-\frac{1}{2}UU_{{{ xxxx}}}
-\frac{3}{2}U_{{x}}U_{{{ xxx}}}-3U_{{x}}V_{{{ xx}}}-\frac{3}{2}U_{{{ xx}}}{U}^{3}+2{U_{{x}}}^{3}\\&+3W_{{{ xx_2}}}-2V_{{{ xx_3}}}+3V_{{{ xxx_2}}}+5UU_{{x}}U_{{x_2}}
+2UVU_{{x_2}}+3U_{{{ xx}}}U_{{x}}U+2U_{{{ xx}}}VU,
\end{split}
\end{equation}
where $U=\varphi_0$, $V=\varphi_1$, $W=\varphi_2$ and $X=\varphi_3$.  This was calculated using Maple and PSEUDO \cite{Brunelli}.  Note that although $X$ and $Y$ appear in this Lagrangian, their presence is trivial in that they do not contribute to or feature in the resulting Euler-Lagrange equations.  We can simplify $\mathscr{L}_{(123)}$ considerably by subtracting total derivatives to obtain the equivalent Lagrangian 

\begin{equation}
\tilde{\mathscr{L}}_{(123)}=3U_{x}^2U^2 - \frac{3}{2}U_{x x_{2}}U^2 + 3V_{x x}U^2  +\frac{5}{2}U_{x}^3  + U_xU_{x_{3}} + U_{xx}^2 - 3U_xV_{ x_{2}} - 3U_xV_{ x x}   +3V_{x}^2
\end{equation}
that gives identical Euler-Lagrange equations.  The variational derivatives with respect to $U$ and $V$ are

\begin{equation}
\begin{split}
\frac{\delta \mathscr{L}_{{(123)}}}{\delta U}=&-6U^2U_{x x} - 6UU_{x}^2 - 6UU_{x x_2} + 6UV_{x x} - 3U_{x}U_{2} - 15U_{x}U_{x x} - 2U_{x x_3} + 2U_{x x x x}\\& + 3V_{x x_2} + 3V_{x x x}\\
\frac{\delta \mathscr{L}_{{(123)}}}{\delta V}=&\,6UU_{xx} + 6U_{x}^2 - 3U_{xxx} + 3U_{xx_2} - 6V_{xx},
\end{split}
\end{equation}
giving us that

\begin{equation}
\begin{split}
\frac{\delta \mathscr{L}_{(123)}}{\delta \phi}=&\partial \frac{\delta \mathscr{L}_{{(123)}}}{\delta V}+\frac{\delta \mathscr{L}_{{(123)}}}{\delta U}\\
=&\frac{\delta \mathscr{L}_{{(123)}}}{\delta V}\partial+ \D_x\frac{\delta \mathscr{L}_{{(123)}}}{\delta V}+\frac{\delta \mathscr{L}_{{(123)}}}{\delta U} \\
=&(6UU_{xx} + 6U_{x}^2 - 3U_{xxx} + 3U_{xx_2} - 6V_{xx})\partial-U_{xxxx} + 6UU_{xxx} + 3U_{xxx_2} - 3V_{xxx}\\& + -6U^2U_{xx} + 3U_{x}U_{xx} - 6UU_{x}^2 - 6UU_{xx_2} + 6UV_{xx} - 3U_{x_2}U_{x} - 2U_{xx_3} + 3V_{xx_2}
\end{split}
\end{equation}

Since the Euler Lagrange equations \eqref{standardEL} have a pre-factor of $\phi^{-1}$, we calculate

\begin{equation}
\begin{split}
\bigg(\phi\frac{\delta \mathscr{L}_{{(123)}}}{\delta \phi}\bigg)_+=&\,(6 U U_{xx}+6 U_{x}^2-3 U_{xxx}+3 U_{xx_2}-6 V_{xx}) \partial-3 U_{x_2} U_{x}-3 U U_{xx_2}\\&+3 U_{xxx_2}+3 V_{xx_2}-2 U_{xx_3}+3 U U_{xxx}+3 U_{x} U_{xx}-U_{xxxx}-3 V_{xxx}.
\end{split}
\end{equation}
Making the substitution $u_1=-U_x$, $u_2=UU_x-V_x$ (based on the expansion \eqref{varphiexp}), this becomes

\begin{equation}
(3u_1^{(2)}-3(u_1)_{x_2}+6u_2^{(1)})\partial + 2(u_1)_{x_3}-3(u_1^{(1)})_{x_2}-3(u_2)_{x_2}-6u_1u_1^{(1)}+u_1^{(3)}+3u_2^{(2)}.
\end{equation}
Setting this equal to zero gives us equations that are equivalent to \eqref{KPsimult}.
\end{example}

\section{Lagrangian multiforms for the KP hierarchy}\label{KPMFSect}

In this section we present two closely related Lagrangian multiform structures for the KP hierarchy.  Let
\begin{equation}\label{multiformM}
\textsf{M}=\sum_{1\leq i<j<k}\mathscr{L}_{(ijk)}\textsf{d}x_i\wedge\textsf{d}x_j\wedge\textsf{d}x_k.
\end{equation}
be a differential 3-form.  We shall define the coefficients $\mathscr{L}_{(ijk)}$ such that the PDEs defined by $\delta\textsf{dM}=0$ are the full set of equations of the KP hierarchy, and we shall show that on these equations $\textsf{dM}=0$.  We define $P_{(ijkl)}$ such that

\begin{equation}\label{1stPdef}
\textsf{dM}=\sum_{1\leq i<j<k<l}{P}_{(ijkl)}\textsf{d}x_i\wedge\textsf{d}x_j\wedge\textsf{d}x_k\wedge\textsf{d}x_l.
\end{equation}
and will show that each $P_{(1ijk)}$ has a double zero on the equations of the KP hierarchy, so the coefficients $P_{(1ijk)}$ will be of the form
 
 \begin{equation}
 \sum_{\gamma=1}^nA_\gamma B_\gamma
 \end{equation}
 where each $A_\gamma$ and $B_\gamma$ is zero on the equations of the KP hierarchy.  More specifically, the $A_\gamma$ will be of the form

\begin{equation}
(L^i_+)_{x_j}-(L^j_+)_{ x_i}+[L^i_+,L^j_+]
\end{equation}
whilst the $B_\gamma$ will be of the form

\begin{equation}
\phi_{x_i}\phi^{-1}+L_-^i,
\end{equation}
giving us the required double zero.   Then
 
 \begin{equation}
\delta P_{(1ijk)}=\sum_{\gamma=1}^n \delta A_\gamma B_\gamma+A_\gamma \delta B_\gamma
\end{equation}
so the equations given by $\delta P_{(1ijk)}=0$ will be a subset of the equations of the KP hierarchy.  In order for the equations given by $\delta P_{(1ijk)}=0$ for all $1<i,j,k$ to be the full set of equations of the KP hierarchy, we require that the factors $A_\gamma$ and $B_\gamma$ span the set of equations of the KP hierarchy, and also that the $A_\gamma$ and $B_\gamma$ are non-degenerate.  Rather than show this directly, we will instead show the equivalent result that the full set of equations of the KP hierarchy arise from the Euler-Lagrange equations of the $\mathscr{L}_{(1ij)}$ Lagrangians.  Then, for the ${P}_{(ijkl)}$ where $1<i,j,k,l$ we will show that $\delta {P}_{(ijkl)}=0$ on the equations of the KP hierarchy.  Together, these results will show that the multiform Euler-Lagrange equations given by $\delta\textsf{dM}=0$ are a subset of the equations of the KP hierarchy, and include the entire KP hierarchy.  It follows that the multiform Euler-Lagrange equations are precisely the equations of the KP hierarchy.\\

The factorised form of $P_{(1ijk)}$ in terms of the $A_\gamma$ and $B_\gamma$ would suggest that as well as giving us equations in the form
\begin{equation}\label{newKPdouble}
(L^i_+)_{x_j}-(L^j_+)_{ x_i}+[L^i_+,L^j_+]=0,
\end{equation}
the multiform Euler-Lagrange equations should also include KP equations of the type
\begin{equation}\label{type3}
\phi_{x_i}\phi^{-1}+L_-^i=0.
\end{equation}
However, Corollary \ref{KPeqcoroll} tells us that the set of equations of the form of \eqref{newKPdouble} for all $i,j>0$ is equivalent to the set of equations of the form of \eqref{type3} for all $i>0$, so we are free to view either of these equivalent sets of equations as the complete set of multiform Euler-Lagrange equations for $\textsf{M}$.\\

\subsection{A Lagrangian multiform for KP based on Dickey's Lagrangian}

We define
\begin{equation}\label{Adef}
\begin{split}
\Gamma_{ijk}:=&\frac{1}{2}([\phi\partial^k\phi^{-1}\phi_{x_i}\phi^{-1}\phi_{x_j},\phi^{-1}]
+[\phi\partial^j\phi^{-1}\phi_{x_k}\phi^{-1}\phi_{x_i},\phi^{-1}]
+[\phi\partial^i\phi^{-1}\phi_{x_j}\phi^{-1}\phi_{x_k},\phi^{-1}]\\
&-[\phi\partial^k\phi^{-1}\phi_{x_j}\phi^{-1}\phi_{x_i},\phi^{-1}]
-[\phi\partial^j\phi^{-1}\phi_{x_i}\phi^{-1}\phi_{x_k},\phi^{-1}]
-[\phi\partial^i\phi^{-1}\phi_{x_k}\phi^{-1}\phi_{x_j},\phi^{-1}]\\
&+[\phi_{x_j},\partial^k\phi^{-1}\phi_{x_i}\phi^{-1}]
+[\phi_{x_i},\partial^j\phi^{-1}\phi_{x_k}\phi^{-1}]
+[\phi_{x_k},\partial^i\phi^{-1}\phi_{x_j}\phi^{-1}]\\
&-[\phi_{x_i},\partial^k\phi^{-1}\phi_{x_j}\phi^{-1}]
-[\phi_{x_k},\partial^j\phi^{-1}\phi_{x_i}\phi^{-1}]
-[\phi_{x_j},\partial^i\phi^{-1}\phi_{x_k}\phi^{-1}]),
\end{split}
\end{equation}

\begin{equation}\label{Bdef}
\begin{split}
\Delta_{ij,k}:=-\int_0^1p^{-1}\big([&T[V,S],U]+[[T,U]_+S,V]+[U[V,S]_+,T]+[UT,[V,S]_+]+[T[S,U],V]\\
&+[U,[T,V]_+S]+[V[S,U]_+,T]+[VT,[S,U]_+]+[[U,V],TS]+[T,[U,V]S]\big)\textsf{d}  p
\end{split}
\end{equation}
where $S=\phi_{ p}^{-1}$, $T=(\phi_{ p})_{x_k}  \phi_{ p}^{-1}$, $U=(\phi_{ p}\partial^i\phi_{ p}^{-1})_+$ and $V=(\phi_{ p}\partial^j\phi_{ p}^{-1})_+$,

\begin{equation}\label{Cdef}
\Theta_{ij,k}:=\frac{1}{2}\big([\phi_{x_k}\phi^{-1},L^i_+L^j_-]+[L^j_-,L^i_+\phi_{x_k}\phi^{-1}]+[L^j_+\phi_{x_k}\phi^{-1},L^i_-]+[L^j_+L^i_-,\phi_{x_k}\phi^{-1}]\big)
\end{equation}
and

\begin{equation}\label{Edef}
\begin{split}
\Lambda_{ijk}:=\frac{1}{2}\big([L^i_+L^j_--L^j_+L^i_-,L^k]+[L^k_+L^i_-,L^j_+]+[L^i_+,L^k_+L^j_-]+[L^i_-,L^{j+k}]+[L^{i+k},L^j_-]\big).
\end{split}
\end{equation}
In these definitions, $L$ is used as an abbreviation of $\phi\partial\phi^{-1}$, so all of the above are pseudodifferential operators whose coefficients are in terms of $\varphi_\beta$ and their derivatives.

\begin{theorem}\label{bigtheorem}
The 3-form

\begin{equation}
\textsf{M}=\sum_{1\leq i<j<k}\mathscr{L}_{(ijk)}\textsf{d}x_i\wedge\textsf{d}x_j\wedge\textsf{d}x_k
\end{equation}
with coefficients

\begin{equation}
\mathscr{L}_{(1jk)}=\res\bigg\{-\int_0^1p^{-1}[(\phi_p\partial^j\phi_p^{-1})_+,(\phi_p\partial^k\phi_p^{-1})_+]\phi_p^{-1}\textsf{d} p+\partial^k\phi^{-1}\phi_{x_j}-\partial^j\phi^{-1}\phi_{x_k}\bigg\}
\end{equation}
and
\begin{equation}\label{commLagrangian}
\mathscr{L}_{(ijk)}=\int\res\big\{\Gamma_{ijk}+\Delta_{ij,k}+\Delta_{jk,i}+\Delta_{ki,j}+\Theta_{ij,k}+\Theta_{jk,i}+\Theta_{ki,j}+\Lambda_{ijk}\big\}\textsf{d}x
\end{equation}
(with the constant of integration set to zero) when $i>1$ is a Lagrangian multiform for the KP hierarchy.  Each $\mathscr{L}_{(ijk)}$ is a local expression in the fields $\varphi_\beta$ and their derivatives.  The multiform Euler-Lagrange equations given by $\delta\textsf{dM}=0$ are the full set of equations of the KP hierarchy and consequences thereof.  On the equations of the KP hierarchy, $\textsf{dM}=0$.
\end{theorem}
We have constructed $\mathscr{L}_{(ijk)}$ in this way so that 

\begin{equation}
\textsf{dM}=\sum_{1\leq i<j<k<l}{P}_{(ijkl)}\textsf{d}x_i\wedge\textsf{d}x_j\wedge\textsf{d}x_k\wedge\textsf{d}x_l.
\end{equation}
 has a double zero on the equations of the KP hierarchy.  In particular, this $\mathscr{L}_{(ijk)}$ is such that

\begin{equation}\label{KPdl1st}
\begin{split}
P_{(1ijk)}=&-\D_{x_k}\mathscr{L}_{(1ij)}-\D_{x_i}\mathscr{L}_{(1jk)}+\D_{x_j}\mathscr{L}_{(1ik)}+\D_{x_1}\mathscr{L}_{(ijk)}\\
=&-\res\big\{\frac{1}{2}((L^i_+)_{x_j}-(L^j_+)_{ x_i}+[L^i_+,L^j_+])(\phi_{x_k}\phi^{-1}+L_-^k)\\
&\quad\quad+\frac{1}{2}((L^j_+)_{x_k}-(L^k_+)_{ x_j}+[L^j_+,L^k_+])(\phi_{x_i}\phi^{-1}+L_-^i)\\
&\quad\quad+\frac{1}{2}((L^k_+)_{x_i}-(L^i_+)_{ x_k}+[L^k_+,L^i_+])(\phi_{x_j}\phi^{-1}+L_-^j)\big\}.
\end{split}
\end{equation}
Before we can show this to be the case, we shall require a number of lemmas.  Lemmas \ref{Alemma} and \ref{Blemma} are closely related to Dickey's computations to obtain the Euler-Lagrange equations of his KP Lagrangian that we reproduced in Section \ref{KPLagSect}.  Lemma \ref{Clemma} then re-arranges some of the resulting terms to get us closer to \eqref{KPdl1st}, whilst Lemma \ref{Elemma} gives us the terms in \eqref{KPdl1st} that do not contain any $x_i$, $x_j$ or $x_k$ derivatives.  Also, it is important to note that each of $\Gamma_{ijk}$, $\Delta_{ij,k}$, $\Theta_{ij,k}$ and $\Lambda_{ijk}$ are expressed in terms of the residue of commutators.  Therefore they are all total $x$ derivatives so can be integrated with respect to $x$ to obtain a local expression for $\mathscr{L}_{(ijk)}$.

\begin{lemma}\label{Alemma}
The $\Gamma_{ijk}$ defined in \eqref{Adef} is such that

\begin{equation}
\begin{split}
&\D_{x_i}(\partial^k\phi^{-1}\phi_{x_j}-\partial^j\phi^{-1}\phi_{x_k})+\D_{x_j}(\partial^i\phi^{-1}\phi_{x_k}-\partial^k\phi^{-1}\phi_{x_i})+\D_{x_k}(\partial^j\phi^{-1}\phi_{x_i}-\partial^i\phi^{-1}\phi_{x_j})\\
&=\frac{1}{2}(-(L^k)_{x_j}\phi_{x_i}+(L^j)_{x_k}\phi_{x_i}-(L^i)_{x_k}\phi_{x_j}+(L^k)_{x_i}\phi_{x_j}-(L^j)_{x_i}\phi_{x_k}+(L^i)_{x_j}\phi_{x_k})\phi^{-1}+\Gamma_{ijk}.
\end{split}
\end{equation}

\end{lemma}

\begin{proof}{of Lemma \ref{Alemma}}
\begin{equation}
\begin{split}
&\D_{x_i}(\partial^k\phi^{-1}\phi_{x_j}-\partial^j\phi^{-1}\phi_{x_k})+\D_{x_j}(\partial^i\phi^{-1}\phi_{x_k}-\partial^k\phi^{-1}\phi_{x_i})+\D_{x_k}(\partial^j\phi^{-1}\phi_{x_i}-\partial^i\phi^{-1}\phi_{x_j})\\
=&\,\partial^k\phi^{-1}\phi_{x_j}\phi^{-1}\phi_{x_i}
+\partial^i\phi^{-1}\phi_{x_k}\phi^{-1}\phi_{x_j}
+\partial^j\phi^{-1}\phi_{x_i}\phi^{-1}\phi_{x_k}\\&
-\partial^k\phi^{-1}\phi_{x_i}\phi^{-1}\phi_{x_j}-\partial^i\phi^{-1}\phi_{x_j}\phi^{-1}\phi_{x_k}-\partial^j\phi^{-1}\phi_{x_k}\phi^{-1}\phi_{x_i}.
\end{split}
\end{equation}
We now use commutators to get this in the form $(L^i)_{x_j}\phi_{x_k}\phi^{-1}$:
\begin{equation}\label{kineticclosure}
\begin{split}
=&\frac{1}{2}(-\phi\partial^k\phi^{-1}\phi_{x_i}\phi^{-1}\phi_{x_j}\phi^{-1}+\phi\partial^j\phi^{-1}\phi_{x_i}\phi^{-1}\phi_{x_k}\phi^{-1}-\phi\partial^i\phi^{-1}\phi_{x_j}\phi^{-1}\phi_{x_k}\phi^{-1}\\
&+\phi\partial^k\phi^{-1}\phi_{x_j}\phi^{-1}\phi_{x_i}\phi^{-1}-\phi\partial^j\phi^{-1}\phi_{x_k}\phi^{-1}\phi_{x_i}\phi^{-1}+\phi\partial^i\phi^{-1}\phi_{x_k}\phi^{-1}\phi_{x_j}\phi^{-1})\\
&+\frac{1}{2}(-\phi_{x_j}\partial^k\phi^{-1}\phi_{x_i}\phi^{-1}+\phi_{x_k}\partial^j\phi^{-1}\phi_{x_i}\phi^{-1}-\phi_{x_k}\partial^i\phi^{-1}\phi_{x_j}\phi^{-1}\\
&+\phi_{x_i}\partial^k\phi^{-1}\phi_{x_j}\phi^{-1}-\phi_{x_i}\partial^j\phi^{-1}\phi_{x_k}\phi^{-1}+\phi_{x_j}\partial^i\phi^{-1}\phi_{x_k}\phi^{-1})\\
&+\frac{1}{2}([\phi\partial^k\phi^{-1}\phi_{x_i}\phi^{-1}\phi_{x_j},\phi^{-1}]
+[\phi\partial^j\phi^{-1}\phi_{x_k}\phi^{-1}\phi_{x_i},\phi^{-1}]
+[\phi\partial^i\phi^{-1}\phi_{x_j}\phi^{-1}\phi_{x_k},\phi^{-1}]\\
&-[\phi\partial^k\phi^{-1}\phi_{x_j}\phi^{-1}\phi_{x_i},\phi^{-1}]
-[\phi\partial^j\phi^{-1}\phi_{x_i}\phi^{-1}\phi_{x_k},\phi^{-1}]
-[\phi\partial^i\phi^{-1}\phi_{x_k}\phi^{-1}\phi_{x_j},\phi^{-1}]\\
&+[\phi_{x_j},\partial^k\phi^{-1}\phi_{x_i}\phi^{-1}]
+[\phi_{x_i},\partial^j\phi^{-1}\phi_{x_k}\phi^{-1}]
+[\phi_{x_k},\partial^i\phi^{-1}\phi_{x_j}\phi^{-1}]\\
&-[\phi_{x_i},\partial^k\phi^{-1}\phi_{x_j}\phi^{-1}]
-[\phi_{x_k},\partial^j\phi^{-1}\phi_{x_i}\phi^{-1}]
-[\phi_{x_j},\partial^i\phi^{-1}\phi_{x_k}\phi^{-1}])\\
=&\frac{1}{2}(-(L^k)_{x_j}\phi_{x_i}+(L^j)_{x_k}\phi_{x_i}-(L^i)_{x_k}\phi_{x_j}+(L^k)_{x_i}\phi_{x_j}-(L^j)_{x_i}\phi_{x_k}+(L^i)_{x_j}\phi_{x_k})\phi^{-1}+\Gamma_{ijk}.
\end{split}
\end{equation}

\end{proof}

\begin{lemma}\label{Blemma}
The $\Delta_{ij,k}$ defined in \eqref{Bdef} is such that
\begin{equation}\label{PEpart}
\begin{split}
&\D_{x_k}\res\bigg\{-\int_0^1p^{-1}[(\phi_p\partial^i\phi_p^{-1})_+,(\phi_p\partial^j\phi_p^{-1})_+]\phi_p^{-1}\textsf{d} p\bigg\}\\
=&\res\big\{[(\phi\partial^i\phi^{-1})_+,(\phi\partial^j\phi^{-1})_+]\phi_{x_k} \ \phi^{-1}\big\}+\res\{ \Delta_{ij,k}\}
\end{split}
\end{equation}

\end{lemma}

\begin{proof}{of Lemma \ref{Blemma}.}
Since each $\mathscr{L}_{(1ij)}$ is autonomous, we notice that $\D_{x_k}\mathscr{L}_{(1ij)}=\delta\mathscr{L}_{(1ij)}|_{\delta\phi=\phi_{x_k}}$.  It follows from Lemma \ref{varprop} that the left hand side of \eqref{PEpart} is equal to

\begin{equation}\label{trivint1}
\res\big\{[(\phi\partial^i\phi^{-1})_+,(\phi\partial^j\phi^{-1})_+]\phi_{x_k} \ \phi^{-1}\big\}-\D_x h_1|_{\delta\phi_{\tilde p}=(\phi_{\tilde p})_{x_k}}
\end{equation}
evaluated at $p=1$.  We note that $\res \{\Delta_{ij,k}\}$ as defined in \eqref{Bdef} is precisely $-\D_x h_1|_{\delta\phi_{\tilde p}=(\phi_{\tilde p})_{x_k}}$ evaluated at $p=1$.  I.e.,

\begin{equation}
\begin{split}
\Delta_{ij,k}:=-\int_0^1p^{-1}\big([&T[V,S],U]+[[T,U]_+S,V]+[U[V,S]_+,T]+[UT,[V,S]_+]+[T[S,U],V]\\
&+[U,[T,V]_+S]+[V[S,U]_+,T]+[VT,[S,U]_+]+[[U,V],TS]+[T,[U,V]S]\big)\textsf{d}  p
\end{split}
\end{equation}
with $S=\phi_{ p}^{-1}$, $T=(\phi_{ p})_{x_k}  \phi_{ p}^{-1}$, $U=(\phi_{ p}\partial^i\phi_{ p}^{-1})_+$ and $V=(\phi_{ p}\partial^j\phi_{ p}^{-1})_+$.
\end{proof}

\begin{lemma}\label{Clemma}
The $\Theta_{ij,k}$ defined in \eqref{Cdef} is such that

\begin{equation}
\begin{split}
\res\{[L^i_+,L^j_+]\phi_{x_k} \phi^{-1}\}=\frac{1}{2}\res\{[L^i_+,L^j_+]\phi_{x_k} \ \phi^{-1}+(L^j_+)_{x_k}L^i_--(L^i_+)_{x_k}L^j_-\}+\res\{\Theta_{ij,k}\}.
\end{split}
\end{equation}

\end{lemma}

\begin{proof}{of Lemma \ref{Clemma}.}
Using the identity
\begin{equation}\label{doublecomm}
0=[L^i,L^j]_+=[L^i_+,L^j_+]+[L^i_+,L^j_-]_++[L^i_-,L^j_+]_+,
\end{equation}
we see that

\begin{equation}
\begin{split}
\res\{[L^i_+,L^j_+]\phi_{x_k} \phi^{-1}\}=&\frac{1}{2}\res\{[L^i_+,L^j_+]\phi_{x_k}  \phi^{-1}\}-\frac{1}{2}\res\{[L^i_+,L^j_-]\phi_{x_k}  \phi^{-1}+[L^i_-,L^j_+]\phi_{x_k}  \phi^{-1}\}\\
=&\frac{1}{2}\res\{[L^i_+,L^j_+]\phi_{x_k}  \phi^{-1}\}+\frac{1}{2}\res\{L^i_+\phi_{x_k}\phi^{-1}L^j_--\phi_{x_k}\phi^{-1}L^i_+L^j_-\\
&+\phi_{x_k}\phi^{-1}L^j_+L^i_- -L^j_+\phi_{x_k}\phi^{-1}L^i_-+[\phi_{x_k}\phi^{-1},L^i_+L^j_-]+[L^j_-,L^i_+\phi_{x_k}\phi^{-1}]\\
&+[L^j_+\phi_{x_k}\phi^{-1},L^i_-]+[L^j_+L^i_-,\phi_{x_k}\phi^{-1}]\}\\
=&\frac{1}{2}\res\{[L^i_+,L^j_+]\phi_{x_k} \ \phi^{-1}+(L^j_+)_{x_k}L^i_--(L^i_+)_{x_k}L^j_-\}\\&+\frac{1}{2}\res\{[\phi_{x_k}\phi^{-1},L^i_+L^j_-]
+[L^j_-,L^i_+\phi_{x_k}\phi^{-1}]+[L^j_+\phi_{x_k}\phi^{-1},L^i_-]\\&+[L^j_+L^i_-,\phi_{x_k}\phi^{-1}]\}\\
=&\frac{1}{2}\res\{[L^i_+,L^j_+]\phi_{x_k} \ \phi^{-1}+(L^j_+)_{x_k}L^i_--(L^i_+)_{x_k}L^j_-\} +\res\{\Theta_{ij,k}\},
\end{split}
\end{equation}
where 

\begin{equation}
\Theta_{ij,k}:=\frac{1}{2}\big([\phi_{x_k}\phi^{-1},L^i_+L^j_-]+[L^j_-,L^i_+\phi_{x_k}\phi^{-1}]+[L^j_+\phi_{x_k}\phi^{-1},L^i_-]+[L^j_+L^i_-,\phi_{x_k}\phi^{-1}]\big).
\end{equation}

\end{proof}

\begin{lemma}\label{Elemma}
The the identity
\begin{equation}
\res\{[L^i_+,L^j_+]L^k_-+[L^j_+,L^k_+]L^i_-+[L^k_+,L^i_+]L^j_-\}=-2\res\{\Lambda_{ijk}\},
\end{equation}
holds.
\end{lemma}

\begin{proof}{of Lemma \ref{Elemma}}
we consider $\res\{[L^i,L^j]L^k\}$, (which is clearly zero) and express this in terms of the positive and negative parts of the powers of $L$:

\begin{equation}\label{finalterm}
\begin{split}
0=\res\{[L^i,L^j]L^k\}=\res\big\{&[L^i_+,L^j_+]L^k_-+[L^i_-,L^j_+]L^k_++[L^i_+,L^j_-]L^k_+\\
+&[L^i_-,L^j_-]L^k_++[L^i_+,L^j_-]L^k_-+[L^i_-,L^j_+]L^k_-\big\}
\end{split}
\end{equation}
The first three terms on the right hand side of \eqref{finalterm} can be written as

\begin{equation}\label{partu}
\begin{split}
\res\big\{&[L^i_+,L^j_+]L^k_-+[L^j_+,L^k_+]L^i_-+[L^k_+,L^i_+]L^j_-\\
+&[L^i_-,L^j_+L^k_+]+[L^k_+,L^j_+L^i_-]+[L^i_+L^j_-,L^k_+]+[L^i_+L^k_+,L^j_-]\big\}
\end{split}
\end{equation}
whilst the final three terms on the right hand side of \eqref{finalterm} can be written as

\begin{equation}
\begin{split}
\res\big\{&\frac{1}{2}([L^j_-,L^k_+]+[L^j_+,L^k_-])L^i_-+\frac{1}{2}([L^k_-,L^i_+]+[L^k_+,L^i_-])L^j_-+\frac{1}{2}([L^i_-,L^j_+]+[L^i_+,L^j_-])L^k_-\\
+&\frac{1}{2}([L^i_-,L^j_-L^k_+]+[L^k_+,L^j_-L^i_-]+[L^i_-L^j_-,L^k_+]+[L^i_-L^k_+,L^j_-]+[L^i_+L^j_-,L^k_-]+[L^i_+L^k_-,L^j_-]\\
+&[L^i_-,L^j_+L^k_-]+[L^k_-,L^j_+L^i_-]\big\}.
\end{split}
\end{equation}
By \eqref{doublecomm}, this is equal to

\begin{equation}\label{partd}
\begin{split}
\frac{1}{2}\res\big\{&-[L^j_+,L^k_+]L^i_--[L^k_+,L^i_+]L^j_--[L^i_+,L^j_+]L^k_-
+[L^i_-,L^j_-L^k_+]+[L^k_+,L^j_-L^i_-]\\
&+[L^i_-L^j_-,L^k_+]+[L^i_-L^k_+,L^j_-]+[L^i_+L^j_-,L^k_-]+[L^i_+L^k_-,L^j_-]
+[L^i_-,L^j_+L^k_-]+[L^k_-,L^j_+L^i_-]\big\}.
\end{split}
\end{equation}
Since \eqref{partu} and \eqref{partd} sum to zero, it follows that

\begin{equation}
\begin{split}
\res&\big\{[L^i_+,L^j_+]L^k_-+[L^j_+,L^k_+]L^i_-+[L^k_+,L^i_+]L^j_-\big\}\\
=&-\res\big\{2[L^i_-,L^j_+L^k_+]+2[L^k_+,L^j_+L^i_-]+2[L^i_+L^j_-,L^k_+]+2[L^i_+L^k_+,L^j_-]+[L^i_-,L^j_-L^k_+]\\&+[L^k_+,L^j_-L^i_-]
+[L^i_-L^j_-,L^k_+]+[L^i_-L^k_+,L^j_-]+[L^i_+L^j_-,L^k_-]+[L^i_+L^k_-,L^j_-]+[L^i_-,L^j_+L^k_-]\\&+[L^k_-,L^j_+L^i_-]\big\}
\end{split}
\end{equation}
which simplifies to

\begin{equation}
\begin{split}
-\res\big\{&[L^i_+L^j_--L^j_+L^i_-,L^k]+[L^k_+L^i_-,L^j_+]+[L^i_+,L^k_+L^j_-]+[L^i_-,L^{j+k}]+[L^{i+k},L^j_-]\big\}\\
=-2\res\{&\Lambda_{ijk}\}
\end{split}
\end{equation}
where

\begin{equation}
\begin{split}
\Lambda_{ijk}:=\frac{1}{2}\big([L^i_+L^j_--L^j_+L^i_-,L^k]+[L^k_+L^i_-,L^j_+]+[L^i_+,L^k_+L^j_-]+[L^i_-,L^{j+k}]+[L^{i+k},L^j_-]\big).
\end{split}
\end{equation}
\end{proof}

%%%%%

\begin{proof}\textit{of Theorem \ref{bigtheorem}.}
Since $\Gamma_{ijk}$, $\Delta_{ij,k}$, $\Theta_{ij,k}$ and $\Lambda_{ijk}$ are composed entirely of commutators, it follows from Lemma \ref{rescomdiv} that

\begin{equation}\label{Lijkdef}
\mathscr{L}_{(ijk)}=\int\res\big\{\Gamma_{ijk}+\Delta_{ij,k}+\Delta_{jk,i}+\Delta_{ki,j}+\Theta_{ij,k}+\Theta_{jk,i}+\Theta_{ki,j}+\Lambda_{ijk}\big\}\textsf{d}x
\end{equation}
is local.  Since the multiform Euler-Lagrange equations arising from $\delta \textsf{dM}=0$ include the Euler-Lagrange equations of the $\mathscr{L}_{1ij}$, we know that the set of equations given by $\delta \textsf{dM}=0$ includes all KP equations of the form

\begin{equation}\label{type1}
(L^i_+)_{x_j}-(L^j_+)_{ x_i}+[L^i_+,L^j_+]=0.
\end{equation}
By Corollary \ref{KPeqcoroll}, $\delta \textsf{dM}=0$ also gives us KP equations of the form

\begin{equation}
\phi_{x_i}+L^i_-\phi=0.
\end{equation}
In order to proceed, we again use the notation $P_{(ijkl)}$ such that

\begin{equation}
\textsf{dM}=\sum_{1\leq i<j<k<l}{P}_{(ijkl)}\textsf{d}x_i\wedge\textsf{d}x_j\wedge\textsf{d}x_k\wedge\textsf{d}x_l.
\end{equation}
Combining the results of Lemmas \ref{Alemma} to \ref{Elemma}, we see that

\begin{equation}\label{KPdL}
\begin{split}
P_{(1ijk)}=&-\D_{x_k}\mathscr{L}_{(1ij)}-\D_{x_i}\mathscr{L}_{(1jk)}+\D_{x_j}\mathscr{L}_{(1ik)}+\D_{x_1}\mathscr{L}_{(ijk)}\\
=&-\res\big\{\frac{1}{2}((L^i_+)_{x_j}-(L^j_+)_{ x_i}+[L^i_+,L^j_+])(\phi_{x_k}\phi^{-1}+L_-^k)\\
&\quad\quad+\frac{1}{2}((L^j_+)_{x_k}-(L^k_+)_{ x_j}+[L^j_+,L^k_+])(\phi_{x_i}\phi^{-1}+L_-^i)\\
&\quad\quad+\frac{1}{2}((L^k_+)_{x_i}-(L^i_+)_{ x_k}+[L^k_+,L^i_+])(\phi_{x_j}\phi^{-1}+L_-^j)\big\}.
\end{split}
\end{equation}
and since equations of the form $(L^i_+)_{x_j}-(L^j_+)_{ x_i}+[L^i_+,L^j_+]=0$ and $\phi_{x_i}\phi^{-1}+L_-^i=0$  are both equations of the KP hierarchy, $P_{1ijk}$ has a double zero on the hierarchy.\\

In order to complete the proof, we must show that for

\begin{equation}\label{Pijkldef}
P_{(ijkl)}=\D_{x_i}\mathscr{L}_{(jkl)}-\D_{x_j}\mathscr{L}_{(ikl)}+\D_{x_k}\mathscr{L}_{(ijl)}-\D_{x_l}\mathscr{L}_{(ijk)},
\end{equation}
$\delta P_{(ijkl)}=0$ and $ P_{(ijkl)}=0$ on the equations of the KP hierarchy.  We require that $\delta P_{(ijkl)}=0$ on the equations of the KP hierarchy in order to confirm that $\delta P_{(ijkl)}=0$ does not define any equations that are not part of the KP hierarchy, and we require that $ P_{(ijkl)}=0$ in order that $\textsf{dM}=0$ on the equations of the hierarchy. To show this, we first note that from its definition in terms of the $\mathscr{L}_{(ijk)}$, $P_{(ijkl)}$ is a polynomial with no constant term, in $(\varphi_\beta^{(n)})_I$ where $n$ gives the order of derivative with respect to $x$ and $I$ is a multi-index representing derivatives with respect to $x_i$ for $i>1$.  Also, since $\textsf{d}^2\textsf{M}$ is identically zero, 

\begin{equation}
D_{x}P_{(ijkl)}=\D_{x_i}P_{(1jkl)}-\D_{x_j}P_{(1ikl)}+\D_{x_k}P_{(1ijl)}-\D_{x_l}P_{(1ijk)}.
\end{equation}
This is an identity, so we do not require the $\varphi_\beta$ to satisfy the equations of the KP hierarchy for this to hold.  Since each of $P_{(1ijk)}$, $P_{(1ikl)}$, $P_{(1ijl)}$, and $P_{(1jkl)}$ has a double zero on the equations of the KP hierarchy, it follows that $\D_x P_{(ijkl)}$ also has a double zero on the equations of the KP hierarchy, and therefore that

\begin{equation}
\frac{\partial }{\partial (\varphi_\beta^{(n)})_{I}}\D_x P_{(ijkl)}=0
\end{equation}
for all $I$ and $n$. Using the identity

\begin{equation}\label{usefulid}
\frac{\partial }{\partial (\varphi_\beta^{(n+1)})_{I}}\D_x P_{(ijkl)}=\D_x\frac{\partial }{\partial (\varphi_\beta^{(n+1)})_{I}} P_{(ijkl)}+\frac{\partial }{\partial (\varphi_\beta^{(n)})_{I}} P_{(ijkl)}
\end{equation}
we see that for a fixed choice of $I$, if $n$ is the largest such that $(\varphi_\beta^{(n)})_{I}$ appears in $P_{(ijkl)}$, then 

\begin{equation}
\frac{\partial }{\partial (\varphi_\beta^{(n)})_{I}} P_{(ijkl)}=0
\end{equation}
on the equations of the KP hierarchy.  It also follows from \eqref{usefulid} that, on the equations of the KP hierarchy,  if
\begin{equation}
\frac{\partial }{\partial (\varphi_\beta^{(n)})_{I}} P_{(ijkl)}=0 \quad \text{then}\quad \frac{\partial }{\partial (\varphi_\beta^{(n-1)})_{I}} P_{(ijkl)}=0.
\end{equation}
Therefore, on the equations of the KP hierarchy,
\begin{equation}\label{allpszero}
\frac{\partial }{\partial (\varphi_\beta^{(n)})_{I}} P_{(ijkl)}=0
\end{equation}
for all $I$ and $n$, so $\delta P_{(ijkl)}=0$.  Since $P_{(ijkl)}$ is autonomous, \eqref{allpszero} tells us that

\begin{equation}
\D_{x_i} P_{(ijkl)}=0 \quad \forall i>0
\end{equation}
so $P_{(ijkl)}$ is constant, and since the KP hierarchy admits the zero solution, we conclude that this constant is zero, and $P_{(ijkl)}=0$ on the equations of the KP hierarchy.\\

Thus, the set of equations defined by $\delta\textsf{dM}=0$ is precisely the full set of equations of the KP hierarchy, and on these equations, $\textsf{dM}=0$, so $\textsf{M}$ is a Lagrangian multiform for the KP hierarchy.
\end{proof}

%%%%%%%%%%%%%%%%%%%%%%%%%%%%%%%%%%%%%%%%%%%%
\subsection{An alternative  KP Lagrangian multiform}

In the KP Lagrangian multiform of Theorem \ref{bigtheorem}, we used Dickey's KP Lagrangian for the $\mathscr{L}_{(1ij)}$ , and the Lagrangian defined in \eqref{commLagrangian} for the $\mathscr{L}_{(ijk)}$ when $1< i,j,k$.  Here we present an alternative version of the KP Lagrangian multiform in which every Lagrangian is of the same type.

\begin{theorem}
The differential 3-form 

\begin{equation}\
\widetilde{\textsf{M}}=\sum_{1\leq i<j <k}\widetilde{\mathscr{L}}_{(ijk)}\ \textsf{d}x_{i}\wedge\textsf{d}x_j\wedge \textsf{d}x_{k}
\end{equation}
where
\begin{equation}
\widetilde{\mathscr{L}}_{(ijk)}=\int\res\big\{\Gamma_{ijk}+\Delta_{ij,k}+\Delta_{jk,i}+\Delta_{ki,j}+\Theta_{ij,k}+\Theta_{jk,i}+\Theta_{ki,j}+\Lambda_{ijk}\big\}\textsf{d}x
\end{equation}
(i.e., the Lagrangian defined in \eqref{commLagrangian}), is a Lagrangian multiform for the KP hierarchy.
\end{theorem}

\begin{proof}
We recall that in Section \ref{KPGD} we identified $x_1$ with $x$.  For now we choose not to do so and treat them as separate co-ordinates.  This allows us to consider a 3-form $\textsf{M}_1$ such that the coefficient of $\textsf{d}x\wedge\textsf{d}x_{i}\wedge\textsf{d}x_{j}$ with $1\leq i<j$ is Dickey's KP Lagrangian $\mathscr{L}_{(xij)}$, whilst the coefficient of $\textsf{d}x_i\wedge\textsf{d}x_{j}\wedge\textsf{d}x_{k}$ with $1\leq i<j<k$ is the Lagrangian $\mathscr{L}_{(ijk)}$ defined in \eqref{commLagrangian}.  It then follows from the proof of Theorem \ref{bigtheorem} that this is also a Lagrangian multiform for the KP hierarchy.  The multiform Euler-Lagrange equations for $\textsf{M}_1$ will be the multiform Euler-Lagrange equations of $\textsf{M}$ plus an additional set of equations that tell us to equate derivatives with respect to $x_1$ with derivatives with respect to $x$, arising from equations of the form

\begin{equation}
(L_+)_{x_j}-(L^j_+)_{ x_1}+[L_+,L^j_+]=0,
\end{equation}
and $\textsf{dM}_1$ will have a double zero on these equations.  We now define $\textsf{M}_2$ to be the restriction of $\textsf{M}_1$ to a submanifold with co-ordinates $x_1,x_2,x_3,\ldots$, obtained by fixing $x=c$, a constant.  
%with all $\mathscr{L}_{(xij)}$ set to zero, 
It follows that $\textsf{dM}_2$ still has a double zero on this same set of equations.  If we then equate $x_1$ with $x$ in $\textsf{M}_2$, we get $\widetilde{\textsf{M}}$ and it follows that $\textsf{d}\widetilde{\textsf{M}}$ has a double zero on the equations of the KP hierarchy.  Therefore, the equations defined by $\delta\textsf{d}\widetilde{\textsf{M}}=0$ are a subset of the equations of the KP hierarchy.\\

To complete the proof that $\widetilde{\textsf{M}}$ is a Lagrangian multiform for the KP hierarchy, we must show that the equations defined by $\delta\textsf{d}\widetilde{\textsf{M}}=0$ are precisely the full set of equations of the KP hierarchy.  We shall do this by showing that the Euler-Lagrange equations of the $\mathscr{L}_{(1jk)}$ Lagrangians give us these equations.\\

We first consider the coefficient $P_{(xijk)}$ from $\textsf{dM}_1$.

\begin{equation}
\begin{split}
P_{(xijk)}=&-\D_{x_k}\mathscr{L}_{(xij)}-\D_{x_i}\mathscr{L}_{(xjk)}+\D_{x_j}\mathscr{L}_{(xik)}+\D_{x}\mathscr{L}_{(ijk)}\\
=&-\res\big\{\frac{1}{2}((L^i_+)_{x_j}-(L^j_+)_{ x_i}+[L^i_+,L^j_+])(\phi_{x_k}\phi^{-1}+L_-^k)\\
&\quad\quad+\frac{1}{2}((L^j_+)_{x_k}-(L^k_+)_{ x_j}+[L^j_+,L^k_+])(\phi_{x_i}\phi^{-1}+L_-^i)\\
&\quad\quad+\frac{1}{2}((L^k_+)_{x_i}-(L^i_+)_{ x_k}+[L^k_+,L^i_+])(\phi_{x_j}\phi^{-1}+L_-^j)\big\},
\end{split}
\end{equation}
so in the case where $i=1$ this becomes

\begin{equation}
\begin{split}
P_{(x1jk)}=&-\D_{x_k}\mathscr{L}_{(x1j)}-\D_{x_1}\mathscr{L}_{(xjk)}+\D_{x_j}\mathscr{L}_{(x1k)}+\D_{x}\mathscr{L}_{(1jk)}\\
=&-\res\big\{\frac{1}{2}(-(L^j_+)_{ x_1}+(L^j_+)_x)(\phi_{x_k}\phi^{-1}+L_-^k)\\
&\quad\quad+\frac{1}{2}((L^j_+)_{x_k}-(L^k_+)_{ x_j}+[L^j_+,L^k_+])(\phi_{x_1}\phi^{-1}+L_-)\\
&\quad\quad+\frac{1}{2}((L^k_+)_{x_1}-(L^k_+)_x)(\phi_{x_j}\phi^{-1}+L_-^j)\big\}
\end{split}
\end{equation}
since $L_+=\partial$.  If we equate $x_1$ and $x$ in this expression then this becomes zero.  This is obvious in the first and third line; for the second line, we note that $L_-=(\phi\partial\phi^{-1})_-=(\partial-\phi_x\phi^{-1})_-=-\phi_x\phi^{-1}$. We now define

\begin{equation}
\bar{\mathscr{L}}_{(xij)}=\mathscr{L}_{(xij)}|_{x\to x_1}
\end{equation}
and consider the 2-form

\begin{equation}
\textsf{L}=\bar{\mathscr{L}}_{(x1j)}\textsf{d}x_{1}\wedge\textsf{d}x_j+\bar{\mathscr{L}}_{(x1k)}\textsf{d}x_{1}\wedge\textsf{d}x_k+(\bar{\mathscr{L}}_{(xjk)}-\bar{\mathscr{L}}_{(1jk)})\textsf{d}x_{j}\wedge\textsf{d}x_k.
\end{equation}
By construction, $\textsf{dL}=-P_{(x1jk)}|_{x\to x_1}=0$.  Then, by Corollary \ref{AppCoroll}, the variational derivative of each of the Lagrangian coefficients in \textsf{L} is zero.  Therefore,

\begin{equation}
\frac{\delta}{\delta \phi}(\bar{\mathscr{L}}_{(xjk)}-\bar{\mathscr{L}}_{(1jk)})=0
\end{equation}
so

\begin{equation}
\frac{\delta \bar{\mathscr{L}}_{(1jk)}}{\delta \phi}=\frac{\delta \bar{\mathscr{L}}_{(xjk)}}{\delta \phi}=\{\phi^{-1}\big((L^i_+)_{x_j}-(L^j_+)_{ x_i}+[L^i_+,L^j_+]\big)\}_+.
\end{equation}
Since $\bar{\mathscr{L}}_{(1jk)}=\widetilde{\mathscr{L}}_{(1jk)}$, all equations of the KP hierarchy are consequences of $\delta\textsf{d}\widetilde{\textsf{M}}=0$, so $\widetilde{\textsf{M}}$ is a Lagrangian multiform for the KP hierarchy.

\end{proof}

%%%%%%%%%%%%%%%%%%%%%%%%%%%%%%%%%%%%%%%%%%%%%%%%%%%%%%
\section{Reduction to multiforms for the Gelfand-Dickey hierarchy}\label{GDmfred}

In order to reduce KP to the $n^{th}$ Gelfand-Dickey hierarchy, we imposed the constraint that $L^n_-=0$.  Since, by \eqref{phiequation}, $\phi_{x_n}=-L^n_-\phi$, we can achieve this in the Lagrangian multiform by setting $\phi_{x_n}=0$.  A simple way to obtain a Lagrangian multiform for the $n^{th}$ Gelfand-Dickey hierarchy is to leave the KP multiform obtained in Section \ref{KPMFSect} unchanged and impose this constraint on the Euler-Lagrange equations.  A more satisfactory approach involves setting $\phi_{x_n}=0$ in \eqref{KPdL} to obtain

\begin{equation}\label{GDdL}
\begin{split}
\D_{x_n}\hat{\mathscr{L}}_{(1ij)}&+\D_{x_i}\hat{\mathscr{L}}_{(1jn)}-\D_{x_j}\hat{\mathscr{L}}_{(1in)}-\D_{x_1}\hat{\mathscr{L}}_{(ijn)}\\
=\res\big\{&\frac{1}{2}((L^i_+)_{x_j}-(L^j_+)_{ x_i}+[L^i_+,L^j_+])L_-^k\\
+&\frac{1}{2}(-(L^n_+)_{x_j}+[L^j_+,L^n_+])(\phi_{x_i}\phi^{-1}+L_-^i)\\
+&\frac{1}{2}((L^n_+)_{ x_i}+[L^n_+,L^i_+])(\phi_{x_j}\phi^{-1}+L_-^j)\big\}.
\end{split}
\end{equation}
If we can find Lagrangians $\hat{\mathscr{L}}_{(ijk)}$ such that \eqref{GDdL} holds, then the constraint $L^n_-=0$ will be naturally incorporated into the multiform Euler-Lagrange equations, giving us the $n^{th}$ Gelfand-Dickey hierarchy.  The $\hat{\mathscr{L}}$ are not uniquely defined by this expression, but a natural choice would be

\begin{subequations}
\begin{equation}
\hat{\mathscr{L}}_{(1ij)}=0,
\end{equation}

\begin{equation}\label{DickeyGD}
\hat{\mathscr{L}}_{(1in)}=\res\bigg\{-\int_0^1p^{-1}[(\phi_p\partial^i\phi_p^{-1})_+,(\phi_p\partial^n\phi_p^{-1})_+]\phi_p^{-1}\textsf{d} p+\partial^n\phi^{-1}\phi_{x_i}\bigg\},
\end{equation}

\begin{equation}
\hat{\mathscr{L}}_{(1jn)}=\res\bigg\{-\int_0^1p^{-1}[(\phi_p\partial^j\phi_p^{-1})_+,(\phi_p\partial^n\phi_p^{-1})_+]\phi_p^{-1}\textsf{d} p+\partial^n\phi^{-1}\phi_{x_j}\bigg\},
\end{equation}
and
\begin{equation}\label{MyGD}
\hat{\mathscr{L}}_{(ijn)}=\int\big\{\hat \Gamma_{ijn}+\Delta_{jn,i}+\Delta_{ni,j}+\Theta_{jn,i}+\Theta_{ni,j}+\Lambda_{ijn}\big\}\textsf{d}x
\end{equation}
\end{subequations}
with the constant of integration set to zero, where

\begin{equation}
\begin{split}
\hat \Gamma_{ijn}=\frac{1}{2}\res\big\{&[\phi\partial^n\phi^{-1}\phi_{x_i}\phi^{-1}\phi_{x_j},\phi^{-1}]-[\phi\partial^n\phi^{-1}\phi_{x_j}\phi^{-1}\phi_{x_i},\phi^{-1}]\\
+&[\phi_{x_j},\partial^n\phi^{-1}\phi_{x_i}\phi^{-1}]-[\phi_{x_i},\partial^n\phi^{-1}\phi_{x_j}\phi^{-1}]\big\}
\end{split}
\end{equation}
is equal to $\Gamma_{ijn}$ with $\phi_{x_n}=0$.  The KP multiform \eqref{multiformM} reduces to

\begin{equation}
\textsf{M}_{(n)}=\sum_{1\leq i<j}\hat{\mathscr{L}}_{(ijn)}\textsf{d}x_i\wedge\textsf{d}x_j\wedge\textsf{d}x_n.
\end{equation}
This multiform does not contain any derivatives with respect to $x_n$, so does not allow any motion in the $x_n$ direction, and is equivalent (i.e., produces identical multiform Euler-Lagrange equations) to

\begin{equation}
\hat {\textsf{M}}_{(n)}=\sum_{1\leq i<j}\hat{\mathscr{L}}_{(ijn)}\textsf{d}x_i\wedge\textsf{d}x_j,
\end{equation}
a Lagrangian 2-form for the $n^{th}$ Gelfand-Dickey hierarchy.  As was the case for the KP Lagrangian multiform,  a Lagrangian multiform with all coefficients in the form of \eqref{MyGD} is also a Lagrangian multiform for the $n^{th}$ Gelfand-Dickey hierarchy.

\section{Conclusion}
The Lagrangian multiforms we have presented constitute, in our view, the first instance of establishing the integrability of the KP hierarchy at the Lagrangian level.  In contrast to the Lagrangian multiform for KP hierarchy (up to the $x_4$ flow) that was presented in \cite{ThisPaper},  we now have explicit formulae for the constituent Lagrangians of the Lagrangian multiform for the complete hierarchy, and the constituent Lagrangians are fully local.  In addition, whilst for the Lagrangian multiform in \cite{ThisPaper} the $x_1$ and $x_2$ co-ordinates held a special status (i.e., were treated differently to the other co-ordinates), for the Lagrangian multiform presented here, only $x_1$ holds a special status.  Aspirations for future work include obtaining a Lagrangian multiform for KP that treats every co-ordinate (including $x$) on an equal footing, and also to connect the continuous KP Lagrangian multiform from this paper with the discrete KP Lagrangian multiform given in \cite{Quispel2009}.

\appendix

\section{Multiform Euler-Lagrange equations in terms of variational derivatives}\label{MFEL}

It was first shown in \cite{Suris2016} that $\delta\textsf{dM}=0$ on critical points of a differential form

\begin{equation}\label{appM}
\textsf{M}=\sum_{1\leq i_1<\ldots <i_k\leq N}\mathscr{L}_{(i_1\ldots i_k)}\ \textsf{d}x_{i_1}\wedge\ldots\wedge \textsf{d}x_{i_k}.
\end{equation}
In \cite{Vermeeren2018} and \cite{ThisPaper}, different proofs are given of how the equations given by $\delta\textsf{dM}=0$ can be expressed in terms of variational derivatives of the coefficients $\mathscr{L}_{(i_1\ldots i_k)}$.  In this section, we shall present an alternative proof of this that also gives explicitly the link between the equations in terms of variational derivatives of the $\mathscr{L}_{(i_1\ldots i_k)}$ and the $P_{(i_1\ldots i_{k+1})}$ defined by

\begin{equation}
\textsf{dM}=\sum_{1\leq i_1<\ldots <i_{k+1}\leq N}P_{(i_1\ldots i_{k+1})} \textsf{d}x_{i_1}\wedge\ldots\wedge \textsf{d}x_{i_{k+1}}.
\end{equation}
In terms of the $\mathscr{L}_{(i_1\ldots i_k)}$, 

\begin{equation}\label{AppP}
P_{(i_1\ldots i_{k+1})}=\sum_{\alpha=1}^{k+1} (-1)^{\alpha+1}\D_{x_{i_\alpha}}\mathscr{L}_{(i_1\ldots i_{\alpha-1}i_{\alpha+1}\ldots i_{k+1})}.
\end{equation}
We recall that the multiform Euler-Lagrange equations are given by $\delta \textsf{dM}=0$.  We introduce the notation $I$ to represent the $N$ component multi-index $(i_1,\ldots,i_N)$ such that
\begin{equation}
u_I:=\bigg(\prod _{\alpha=1}^p(\D_{x_\alpha})^{i_\alpha}\bigg)u.
\end{equation}
We shall write $Ik^r$ to denote $(i_1,\ldots,i_k+r,\ldots,i_N)$, $I\backslash k^r$ to denote $(i_1,\ldots,i_k-r,\ldots,i_N)$ and $|I|$ to denote the sum $ i_1+\ldots +i_N$.  This allows us to express the multiform Euler-Lagrange equations are given by $\delta \textsf{dM}=0$ in the form

\begin{equation}\label{pcondition}
\frac{\partial}{\partial u_I}P_{(i_1\ldots i_{k+1})}=0
\end{equation}
for all $1\leq i_1<\ldots <i_{k+1}$ and all multi-indices $I$.  For a fixed choice of $i_1\ldots i_{k+1}$, we shall write $\mathscr{L}_{(\bar \alpha)}$ to denote $\mathscr{L}_{(i_1\ldots i_{\alpha-1}i_{\alpha+1}\ldots i_{k+1})}$.  We then define

\begin{equation}\label{vardivdef}
\frac{\delta\mathscr{L}_{(\bar \alpha)}}{\delta u_{I}}=\sum_{\substack{  J\\j_{i_\alpha}=0}}(-\D)_J\frac{\partial\mathscr{L}_{(\bar \alpha)}}{\partial u_{IJ}},
\end{equation}
where the multi-index $J$ is such that components $j_\alpha=0$ whenever $\alpha\neq i_1,\ldots, i_{k+1}$, i.e. $J$ represents derivatives with respect to $x_{i_1}, \ldots,x_{i_{k+1}}$ only.  We define that $\dfrac{\delta\mathscr{L}_{(\bar \alpha)}}{\delta u_{I}}=0$ in the case where any component of the multi-index $I$ is negative.  Note that by this definition, the variational derivative of the Lagrangian $\mathscr{L}_{(i_1\ldots i_{\alpha-1}i_{\alpha+1}\ldots i_{k+1})}$ with respect to $u_I$ only sees derivatives of $u_I$ with respect to the variables $x_{i_1},\ldots ,x_{ i_{\alpha-1}},x_{i_{\alpha+1}},\ldots, x_{i_{k+1}}$, even though derivatives with respect to other variables may appear in the Lagrangian.  This corresponds with only being able to perform integration by parts with respect to variables that are integrated over in the action.\\

Using the identity

\begin{equation}
\frac{\partial}{\partial u_I}\D_{x_i}=\frac{\partial}{\partial u_{I\backslash i}}+\D_{x_i}\frac{\partial}{\partial u_I}
\end{equation}
tells us that

\begin{equation}
\frac{\partial}{\partial u_I}P_{(i_1\ldots i_{k+1})}=\sum_{\alpha=1}^{k+1}(-1)^{\alpha+1}\bigg(\frac{\partial \mathscr{L}_{(\bar\alpha)}}{\partial u_{I\backslash i_\alpha}}+\D_{x_{i_\alpha}}\frac{\partial \mathscr{L}_{(\bar\alpha)}}{\partial u_{I}}\bigg)
\end{equation}
so

\begin{equation}\label{1stdeltaP}
\begin{split}
\frac{\delta}{\delta u_I}P_{(i_1\ldots i_{k+1})}&=\sum_J(-\D)_J\frac{\partial}{\partial u_{IJ}}P_{(i_1\ldots i_{k+1})}\\&=\sum_J(-\D)_J\sum_{\alpha=1}^{k+1}(-1)^{\alpha+1}\bigg(\frac{\partial \mathscr{L}_{(\bar\alpha)}}{\partial u_{IJ\backslash i_\alpha}}+\D_{x_{i_\alpha}}\frac{\partial \mathscr{L}_{(\bar\alpha)}}{\partial u_{IJ}}\bigg).
\end{split}
\end{equation}
Whenever $j_{i_\alpha}\neq 0$ in this sum, so $J$ is of the form $Ki_\alpha$ for some multi-index $K$, then

\begin{equation}
\pm(-\D)_J\frac{\partial \mathscr{L}_{(\bar\alpha)}}{\partial u_{IJ\backslash i_\alpha}}=\mp\D_{x_{i_\alpha}}(-\D)_K\frac{\partial \mathscr{L}_{(\bar\alpha)}}{\partial u_{IK}}
\end{equation}
will appear in this sum.  When $J=K$, the term

\begin{equation}
\pm(-\D)_K\D_{x_{i_\alpha}}\frac{\partial \mathscr{L}_{(\bar\alpha)}}{\partial u_{IK}}
\end{equation}
will appear.  These two terms cancel, so \eqref{1stdeltaP} simplifies to

\begin{equation}
\begin{split}
\frac{\delta}{\delta u_I}P_{(i_1\ldots i_{k+1})}&=\sum_{\alpha=1}^{k+1}\,\sum_{\substack{ J\\j_{i_\alpha}=0}}(-1)^{\alpha+1}(-\D)_J\frac{\partial \mathscr{L}_{(\bar\alpha)}}{\partial u_{IJ\backslash i_\alpha}}\\&=\sum_{\alpha=1}^{k+1}(-1)^{\alpha+1}\frac{\delta\mathscr{L}_{(\bar \alpha)}}{\delta u_{I\backslash i_{\alpha}}}.
\end{split}
\end{equation}
It follows that if \eqref{pcondition} holds, then

\begin{equation}
\frac{\delta}{\delta u_I}P_{(i_1\ldots i_{k+1})}=\sum_{\alpha=1}^{k+1}(-1)^{\alpha+1}\frac{\delta\mathscr{L}_{(\bar \alpha)}}{\delta u_{I\backslash i_{\alpha}}}=0.
\end{equation}
We have shown that 

\begin{equation}
\delta\textsf{dM}=0\implies \frac{\delta}{\delta u_I}P_{(i_1\ldots i_{k+1})}=\sum_{\alpha=1}^{k+1}(-1)^{\alpha+1}\frac{\delta\mathscr{L}_{(\bar \alpha)}}{\delta u_{I\backslash i_{\alpha}}}=0
\end{equation}
for all $1\leq i_1\leq\ldots\leq i_{k+1}\leq N$ and $I$.  Since

\begin{equation}
\frac{\partial P_{(i_1\ldots i_{k+1})}}{\partial u_I}=\sum_{\substack{  J\\j_i\leq 1}}\D_J\frac{\delta P_{(i_1\ldots i_{k+1})}}{\delta u_{IJ}}
\end{equation}
(a proof of this identity is given in \cite{ThisPaper}) it follows that the converse also holds.  We summarise this result in the following theorem:

\begin{theorem}\label{Apptheorem}
For a differential $k$-form \textsf{M} as given in \eqref{appM}, and $P_{(i_1\ldots i_{k+1})}$ as defined in \eqref{AppP},

\begin{equation}
\frac{\delta}{\delta u_I}P_{(i_1\ldots i_{k+1})}=\sum_{\alpha=1}^{k+1}(-1)^{\alpha+1}\frac{\delta\mathscr{L}_{(\bar \alpha)}}{\delta u_{I\backslash i_{\alpha}}}.
\end{equation}
The set of equations defined by

\begin{equation}
\frac{\delta}{\delta u_I}P_{(i_1\ldots i_{k+1})}=0
\end{equation}
for all $1\leq i_1\leq\ldots\leq i_{k+1}\leq N$ and $I$ is equivalent to the set of equations defined by $\delta\textsf{dM}=0$. 
\end{theorem}

\begin{corollary}\label{AppCoroll}
A corollary of Theorem \ref{Apptheorem} is that

\begin{equation}
\frac{\delta}{\delta u_{x_{i_\alpha}}}P_{(i_1\ldots i_{k+1})}=(-1)^{\alpha+1}\frac{\delta\mathscr{L}_{(i_1\ldots i_{\alpha-1}i_{\alpha+1}\ldots i_{k+1})}}{\delta u},
\end{equation}
so the usual Euler-Lagrange equations of each Lagrangian coefficient in $M$ can be expressed in terms of variational derivatives of the coefficients of $\textsf{dM}$.
\end{corollary}

\section{Explicit form of the KP Lagrangian multiform}\label{longLagrangians}

Here we present the first four Lagrangians of the KP Lagrangian multiform \textsf{M} and $\tilde{\textsf{M}}$, expressed in terms of the $\varphi_\beta$ that constitute $\phi$.  In order to avoid notational confusion over the use of subscripts, we let $U=\varphi_0$, $V=\varphi_1$, $W=\varphi_2$ and $X=\varphi_3$.  The following Lagrangians were found using Maple and PSEUDO \cite{Brunelli}.  In order to obtain $\mathscr{L}_{(234)}$, a Maple procedure based on \eqref{hproc} was used. 

\begin{dmath}
\mathscr{L}_{{(123)}} =-U_{{{ xxx_3}}}+X_{{x_2}}-VU_{{{ xx_2}}}-WU_{{x_2}}-VV_{{x_2}}-{U}^{2}U_{{x_3}}+VU_{{x_3}}+UU_{{{ xx_3}}}
+{U}^{2}U_{{{ xx_2}}}+UV_{{x_3}}+{U}^{2}V_{{x_2}}-UU_{{{ xxx_2}}}-{U}^{3}U_{{x_2}}-UW_{{x_2}}-2UV_{{{ xx_2}}}-3V_{{x}}U_{{x_2}}-3U_{{{ xx}}}U_{{x_2}}
+2U_{{x}}U_{{x_3}}-3U_{{x}}V_{{x_2}}-3U_{{x}}U_{{{ xx_2}}}-W_{{x_3}}+U_{{{ xxxx_2}}}-\frac{3}{2}UV_{{{ xxx}}}
-\frac{3}{2}U_{{{ xxx}}}V-3V_{{{ xx}}}V-\frac{3}{2}{U_{{x}}}^{2}{U}^{2}+2U_{{{ xxx}}}{U}^{2}+2V_{{{ xx}}}{U}^{2}+2{U_{{x}}}^{2}V-\frac{1}{2}UU_{{{ xxxx}}}
-\frac{3}{2}U_{{x}}U_{{{ xxx}}}-3U_{{x}}V_{{{ xx}}}-\frac{3}{2}U_{{{ xx}}}{U}^{3}+2{U_{{x}}}^{3}+3W_{{{ xx_2}}}-2V_{{{ xx_3}}}+3V_{{{ xxx_2}}}+5UU_{{x}}U_{{x_2}}
+2UVU_{{x_2}}+3U_{{{ xx}}}U_{{x}}U+2U_{{{ xx}}}VU,
\end{dmath}

\begin{dmath}
\widetilde{\mathscr{L}}_{{(123)}} =2U^2U_{x x x} + 3UU_{x}U_{x x} + 2U_{x}^3 + \frac{1}{2}U_{x_2}V_{x} - \frac{1}{2}U_{x}V_{x_2} - 2U^2U_{x x_2} + \frac{3}{2}VU_{x x_2} + \frac{3}{2}UU_{x x x_2} + 2U_{x}U_{x x_2} - \frac{3}{2}VU_{x x x} - \frac{3}{2}UV_{x x x} - \frac{3}{2}U^3U_{x x} - 3U_{x}V_{x x} - \frac{3}{2}U_{x}U_{x x x} - \frac{1}{2}U_{x_2}U_{x x} - \frac{1}{2}UU_{x x x x} - UU_{x}U_{x_2} - UU_{x x_3} + 2VU_{x}^2 + 2U^2V_{x x} - 3VV_{x x} - \frac{3}{2}U^2U_{x}^2 + \frac{3}{2}UV_{x x_2} + 2UU_{x x}V,
\end{dmath}

\begin{dmath}
\mathscr{L}_{{(134)}} = -6U_{{{ xx}}}V_{{{ xxx}}}-\frac{3}{2}U_{{x}}U_{{{ xxxxx}}}
-5U_{{x}}V_{{{ xxxx}}}-6U_{{x}}W_{{{ xxx}}}-4V_{{x}}U_{{{ xxxx}}}+U_{{{ xxxxx_2}}}
+40V_{{x}}U_{{x}}U_{{{ xx}}}-6WV_{{{ xxx}}}-12V_{{x}}V_{{{ xxx}}}-4U_{{x}}W_{{x_2}}+Y_{{x_2}}+UW_{{x_3}}-4V_{{x_2}}V_{{x}}
-6V_{{x_2}}U_{{{ xx}}}+8U_{{x_2}}{U_{{x}}}^{2}-4U_{{x_2}}W_{{x}}-6U_{{x_2}}V_{{{ xx}}}-4U_{{x_2}}U_{{{ xxx}}}+\frac{14}{3}{U}^{2}V_{{{ xxxx}}}
+2{U}^{2}U_{{{ xxxxx}}}-2{U}^{5}U_{{{ xx}}}+{\frac {96}{5}{U}^{2}{U_{{x}}}^{3}}+{\frac {12}{5}{U}^{4}V_{{{ xx}}}}+{\frac {24}{5}{U}^{4}U_{{{ xxx}}}}-4{U}^{4}{U_{{x}}}^{2}-\frac{21}{2}{U}^{2}{U_{{{ xx}}}}^{2}
-6{U}^{3}V_{{{ xxx}}}-U_{{{ xxxx_3}}}-3{U}^{3}W_{{{ xx}}}-6W_{{{ xx}}}W-6U_{{{ xx}}}W_{{{ xx}}}-2U_{{{ xx}}}U_{{{ xxxx}}}-\frac{3}{2}UV_{{{ xxxxx}}}
-\frac{9}{2}{U}^{3}U_{{{ xxxx}}}+UU_{{{ xxx_3}}}+2UV_{{{ xx_3}}}+V_{{x_3}}V+U_{{{ xx_3}}}V+WU_{{x_3}}-U_{{{ xx_3}}}{U}^{2}
-V_{{x_3}}{U}^{2}-3UW_{{{ xx_2}}}+U_{{x_3}}{U}^{3}-UX_{{x_2}}-VU_{{{ xxx_2}}}-U_{{{ xx_2}}}W+{U}^{2}U_{{{ xxx_2}}}
-8U_{{x}}V_{{{ xx_2}}}-4U_{{{ xx_2}}}V_{{x}}-6U_{{{ xx_2}}}U_{{{ xx}}}-4U_{{x}}U_{{{ xxx_2}}}+3U_{{x_3}}V_{{x}}+3U_{{x_3}}U_{{{ xx}}}+3V_{{x_3}}U_{{x}}
+3U_{{{ xx_3}}}U_{{x}}-3UV_{{{ xxx_2}}}-UU_{{{ xxxx_2}}}+{U}^{2}W_{{x_2}}-VW_{{x_2}}-V_{{x_2}}{U}^{3}-V_{{x_2}}W+U_{{x_2}}{U}^{4}
+U_{{x_2}}{V}^{2}-U_{{x_2}}X+2{U}^{2}V_{{{ xx_2}}}-2VV_{{{ xx_2}}}-U_{{{ xx_2}}}{U}^{3}+{\frac {24}{5}{U}^{3}U_{{x}}V_{{x}}}+24{U}^{3}U_{{x}}U_{{{ xx}}}
-5U_{{x_3}}U_{{x}}U-2U_{{x_3}}UV-12V_{{x}}W_{{{ xx}}}+20U_{{x}}{U_{{{ xx}}}}^{2}+16U_{{x}}{V_{{x}}}^{2}+{\frac {34}{3}{U_{{x}}}^{2}U_{{{ xxx}}}}
+8{U_{{x}}}^{2}V_{{{ xx}}}-2UW_{{{ xxxx}}}-3{U_{{x}}}^{4}+2U_{{{ xx_2}}}UV+7U_{{{ xx_2}}}U_{{x}}U+9U_{{x_2}}UU_{{{ xx}}}+7U_{{x_2}}UV_{{x}}
+2U_{{x_2}}WU+6U_{{x_2}}U_{{x}}V-9U_{{x_2}}U_{{x}}{U}^{2}-3U_{{x_2}}{U}^{2}V-X_{{x_3}}+16UU_{{x}}W_{{{ xx}}}+{\frac {46}{3}U_{{{ xxx}}}U_{{x}}V}
+7V_{{x_2}}U_{{x}}U+2V_{{x_2}}UV+{\frac {70}{3}UU_{{x}}V_{{{ xxx}}}}+8U_{{x}}VV_{{{ xx}}}+{\frac {41}{3}UU_{{x}}U_{{{ xxxx}}}}+4U_{{{ xx}}}WV
+12U_{{x}}U_{{{ xx}}}W-12UVU_{{x}}V_{{x}}-42UVU_{{x}}U_{{{ xx}}}-6VW_{{{ xxx}}}-2U_{{{ xxxx}}}W+12{U_{{{ xx}}}}^{2}V+6UU_{{{ xx}}}V_{{{ xx}}}
+12UU_{{{ xxx}}}U_{{{ xx}}}-60U{U_{{x}}}^{2}U_{{{ xx}}}+8U_{{x}}WV_{{x}}+16U_{{{ xx}}}VV_{{x}}+4UV_{{x}}V_{{{ xx}}}+{\frac {28}{3}UU_{{{ xxx}}}V_{{x}}}
-33U{U_{{x}}}^{2}V_{{x}}+12UVV_{{{ xxx}}}-\frac{1}{2}UU_{{{ xxxxxx}}}+4UV_{{{ xx}}}W+{\frac {22}{3}UVU_{{{ xxxx}}}}+8UVW_{{{ xx}}}
+4UU_{{{ xxx}}}W-6U{V}^{2}U_{{{ xx}}}-6U{U_{{x}}}^{2}W-27{U}^{2}U_{{{ xxx}}}U_{{x}}+{\frac {36}{5}{U}^{3}VU_{{{ xx}}}}+{\frac {48}{5}{U}^{2}V{U_{{x}}}^{2}}+6W_{{{ xxx_2}}}
+4V_{{{ xxxx_2}}}-3W_{{{ xx_3}}}-3V_{{{ xxx_3}}}+4X_{{{ xx_2}}}-9{U}^{2}V_{{x}}U_{{{ xx}}}-6{U}^{2}V_{{{ xx}}}V-3{U}^{2}WU_{{{ xx}}}-15{U}^{2}V_{{{ xx}}}U_{{x}}
-12{U}^{2}U_{{{ xxx}}}V+4{U}^{2}W_{{{ xxx}}}-3{U_{{x}}}^{3}V+4U_{{{ xxx}}}{V}^{2}-5VV_{{{ xxxx}}}-\frac{3}{2}VU_{{{ xxxxx}}},
\end{dmath}

\begin{dmath}
\widetilde{\mathscr{L}}_{{(134)}} = -3U_{x}^3V - 4U_{x}^2U^4 + 16U_{x x}V_{x}V - 5VV_{x x x x} + 2UU_{x x x x_3} + 8UVW_{x x} - 6VW_{x x x} - 6U_{x}W_{x x x} - 6U_{x}^2UW - \frac{9}{2}U^3U_{x x x x} - 6U^3V_{x x x} + 2U_{x x_3}W + 24U_{x x}U_{x}U^3 - 2V_{x_3}U_{x x} + \frac{24}{5}U_{x x x}U^4 + \frac{28}{3}UV_{x}U_{x x x} - 6U_{x x}V_{x x x} + 8U_{x}^2V_{x x} + 16U_{x}UW_{x x} - 3U^3W_{x x} - 2U_{x x}U_{x x x x} - 3U^2WU_{x x} - \frac{3}{2}VU_{x x x x x} + 20U_{x}U_{x x}^2 - 6U_{x x}W_{x x} - 2UW_{x x x x} - U_{x}W_{x_3} + 2U^2U_{x x x x x} + \frac{24}{5}U_{x}V_{x}U^3 - 42U_{x x}U_{x}UV - 2U_{x}U_{x x_4} + 3U^3U_{x x_3} + 3UV_{x x x_3} + 3VU_{x x x_3} + 4U_{x x x_3}U_{x} + 4U^2W_{x x x} - 4U^2U_{x x x_3} + 2V_{x x_3}U_{x} + 2VV_{x x_3} - \frac{3}{2}U_{x}U_{x x x x x} - U_{x x}UU_{x_3} - 2WU_{x x x x} + 2U^2U_{x x_4} - 12U_{x}UVV_{x} + 6U_{x x_3}V_{x} - 2U^5U_{x x} + 16V_{x}^2U_{x} + 2UW_{x x_3} + 2U_{x x_3}U_{x x} + 4U_{x x x}V^2 + 12U_{x x}^2V + U_{x_3}W_{x} - 12V_{x}W_{x x} - 4U_{x x x x}V_{x} - \frac{1}{2}V_{x}U_{x_4} + 8U_{x}VV_{x x} - 27U^2U_{x}U_{x x x} + 12V_{x x x}UV + \frac{14}{3}U^2V_{x x x x} + \frac{96}{5}U^2U_{x}^3 + 4UWV_{x x} + \frac{46}{3}U_{x x x}U_{x}V - U_{x_3}U_{x x x} + 12U_{x x}U_{x}W - 5U_{x}V_{x x x x} - 33UV_{x}U_{x}^2 + \frac{22}{3}U_{x x x x}UV - 6WV_{x x x} - \frac{21}{2}U^2U_{x x}^2 - 60U_{x}^2UU_{x x} + UU_{x}U_{x_4} + 3U^2U_{x}U_{x_3} + 8U_{x}V_{x}W + \frac{34}{3}U_{x}^2U_{x x x} - \frac{3}{2}UV_{x x_4} - \frac{3}{2}UU_{x x x_4} + 4U_{x x x}UW - \frac{7}{3}UU_{x_3}V_{x} - 12V_{x x x}V_{x} + 4UV_{x}V_{x x} - \frac{16}{3}UU_{x x_3}V - \frac{3}{2}VU_{x x_4} + \frac{70}{3}U_{x}V_{x x x}U - \frac{4}{3}VU_{x}U_{x_3} - 12U_{x x x}VU^2 - 3U_{x}^4 - 6U^2VV_{x x} - 15U^2U_{x}V_{x x} + \frac{48}{5}U_{x}^2VU^2 - \frac{35}{3}UU_{x}U_{x x_3} - \frac{1}{3}UU_{x}V_{x_3} - 6UV^2U_{x x} - \frac{4}{3}U_{x}^2U_{x_3} - 9U_{x x}V_{x}U^2 - \frac{8}{3}U^2V_{x x_3} + \frac{36}{5}U^3VU_{x x} + 40U_{x}U_{x x}V_{x} - 6WW_{x x} + 4VWU_{x x} + 12UU_{x x x}U_{x x} + 6U_{x x}UV_{x x} + \frac{41}{3}UU_{x}U_{x x x x} + \frac{12}{5}U^4V_{x x} + \frac{1}{2}U_{x x}U_{x_4} - \frac{1}{2}UU_{x x x x x x} + \frac{1}{2}U_{x}V_{x_4} - \frac{3}{2}UV_{x x x x x},
\end{dmath}

\begin{dmath}
\mathscr{L}_{{(142)}} = 6{U}^{3}U_{{{ xxx}}}+4{U}^{3}V_{{{ xx}}}-{\frac {24}{5}{U}^{3}{U_{{x}}}^{2}}-{\frac {16}{5}{U}^{4}U_{{{ xx}}}}+2U_{{{ xx}}}U_{{{ xxx}}}
-U_{{{ xxxxx_2}}}+4U_{{{ xx}}}V_{{{ xx}}}-16VU_{{{ xx}}}U_{{x}}-{\frac {20}{3}UU_{{{ xx}}}V_{{x}}}-\frac{16}{3}VV_{{x}}U_{{x}}
-16UU_{{{ xxx}}}U_{{x}}-{\frac {44}{3}UV_{{{ xx}}}U_{{x}}}+4U_{{x}}W_{{x_2}}+U_{{{ xxx_3}}}+{U}^{2}U_{{x_3}}-VU_{{x_3}}
-UU_{{{ xx_3}}}-UV_{{x_3}}-2U_{{x}}U_{{x_3}}+W_{{x_3}}-Y_{{x_2}}+4V_{{x_2}}V_{{x}}+6V_{{x_2}}U_{{{ xx}}}-8U_{{x_2}}{U_{{x}}}^{2}
+4U_{{x_2}}W_{{x}}+6U_{{x_2}}V_{{{ xx}}}+4U_{{x_2}}U_{{{ xxx}}}+2V_{{{ xx_3}}}+3VU_{{{ xxxx}}}+8VV_{{{ xxx}}}+4VW_{{{ xx}}}
-\frac{8}{3}W{U_{{x}}}^{2}+12U{U_{{x}}}^{3}-6U{U_{{{ xx}}}}^{2}+4V_{{{ xx}}}W-4{U}^{2}U_{{{ xxxx}}}-{\frac {20}{3}{U}^{2}V_{{{ xxx}}}}-\frac{8}{3}{U}^{2}W_{{{ xx}}}+2U_{{{ xxx}}}W
-\frac{8}{3}{V}^{2}U_{{{ xx}}}-\frac{8}{3}UU_{{{ xx}}}W-{\frac {28}{3}UU_{{{ xxx}}}V}-8UV_{{{ xx}}}V+8UV{U_{{x}}}^{2}+4{U}^{2}V_{{x}}U_{{x}}+3UW_{{{ xx_2}}}
+UX_{{x_2}}+VU_{{{ xxx_2}}}+U_{{{ xx_2}}}W-{U}^{2}U_{{{ xxx_2}}}+8U_{{x}}V_{{{ xx_2}}}+4U_{{{ xx_2}}}V_{{x}}+6U_{{{ xx_2}}}U_{{{ xx}}}
+4U_{{x}}U_{{{ xxx_2}}}+3UV_{{{ xxx_2}}}+UU_{{{ xxxx_2}}}-{U}^{2}W_{{x_2}}+VW_{{x_2}}+V_{{x_2}}{U}^{3}+V_{{x_2}}W-U_{{x_2}}{U}^{4}
-U_{{x_2}}{V}^{2}+U_{{x_2}}X-2{U}^{2}V_{{{ xx_2}}}+2VV_{{{ xx_2}}}+U_{{{ xx_2}}}{U}^{3}+3U_{{x}}U_{{{ xxxx}}}+4U_{{x}}W_{{{ xx}}}+8U_{{x}}V_{{{ xxx}}}
+4U_{{{ xxx}}}V_{{x}}+8V_{{{ xx}}}V_{{x}}-{\frac {32}{3}V_{{x}}{U_{{x}}}^{2}}-16{U_{{x}}}^{2}U_{{{ xx}}}-2U_{{{ xx_2}}}UV-7U_{{{ xx_2}}}U_{{x}}U
-9U_{{x_2}}UU_{{{ xx}}}-7U_{{x_2}}UV_{{x}}-2U_{{x_2}}WU-6U_{{x_2}}U_{{x}}V+9U_{{x_2}}U_{{x}}{U}^{2}+3U_{{x_2}}{U}^{2}V
-7V_{{x_2}}U_{{x}}U-2V_{{x_2}}UV+UU_{{{ xxxxx}}}+3UV_{{{ xxxx}}}-6W_{{{ xxx_2}}}-4V_{{{ xxxx_2}}}-4X_{{{ xx_2}}}+2UW_{{{ xxx}}}
+22{U}^{2}U_{{{ xx}}}U_{{x}}+8{U}^{2}VU_{{{ xx}}},
\end{dmath}

\begin{dmath}
\widetilde{\mathscr{L}}_{{(142)}} = 6U^3U_{x x x} + \frac{1}{3}UU_{x}V_{x_2} + \frac{7}{3}UU_{x_2}V_{x} - \frac{16}{3}U_{x}V_{x}V - \frac{20}{3}U_{x x}UV_{x} - 2UU_{x x x x_2} - \frac{8}{3}V^2U_{x x} + 8U^2VU_{x x} - 2U_{x x_2}W + 2V_{x_2}U_{x x} - 16U_{x}^2U_{x x} + \frac{8}{3}U^2V_{x x_2} + 3V_{x x x x}U + 2W_{x x x}U + \frac{4}{3}VU_{x}U_{x_2} - \frac{8}{3}U_{x}^2W + 4U^3V_{x x} + 12UU_{x}^3 + U_{x}W_{x_2} + \frac{4}{3}U_{x}^2U_{x_2} + 4WV_{x x} + 3VU_{x x x x} + 8VV_{x x x} + 4VW_{x x} - 3U^3U_{x x_2} - 3UV_{x x x_2} - 3VU_{x x x_2} - 4U_{x x x_2}U_{x} + 4U^2U_{x x x_2} - 2V_{x x_2}U_{x} - 2VV_{x x_2} + U_{x x}UU_{x_2} - 6U_{x x_2}V_{x} - 2UW_{x x_2} - 2U_{x x_2}U_{x x} + 2U_{x x}U_{x x x} - U_{x_2}W_{x} + 8U_{x}V_{x x x} - 4U^2U_{x x x x} + 2WU_{x x x} + 4U_{x x}V_{x x} + 4U_{x x x}V_{x} + U_{x_2}U_{x x x} + 3U_{x}U_{x x x x} + 8V_{x}V_{x x} - 8UVV_{x x} - 16U_{x x}U_{x}V - 3U^2U_{x}U_{x_2} - \frac{32}{3}U_{x}^2V_{x} - \frac{8}{3}UWU_{x x} + \frac{35}{3}UU_{x}U_{x x_2} - \frac{28}{3}U_{x x x}UV - 16U_{x x x}UU_{x} + UU_{x x_4} + 4U_{x}W_{x x} + \frac{16}{3}UU_{x x_2}V - \frac{44}{3}UU_{x}V_{x x} - \frac{24}{5}U^3U_{x}^2 + 22U^2U_{x}U_{x x} + 4U_{x}V_{x}U^2 - \frac{16}{5}U^4U_{x x} + UU_{x x x x x} - \frac{8}{3}U^2W_{x x} - 6U_{x x}^2U - \frac{20}{3}U^2V_{x x x} + 8UU_{x}^2V,
\end{dmath}

\begin{dmath*}
\mathscr{L}_{(234)} = -12U_{x_2}UVU_{xx} - 9U^2U_{xx_2}V_{x} + 14U_{x}V_{x}V_{xx} + 4UU_{x}W_{xxx} + U_{x}VV_{xxx} + 2UU_{xx}^2V + 3UU_{x}V_{xxxx} + UU_{x}U_{xxxxx} + 12U_{x}V_{x}W_{x} + 4U_{x}U_{xx}U_{xxx} + 6U_{x}VW_{xx} + UU_{xx_2}U_{x_3} + U_{x_2}V_{x_3}U - U_{x_2}U_{xx_3}U - UV_{x_2}U_{x_3} - U_{xxxx_2}U_{xx} + \frac{14}{3}U^2V_{xxxx_2} - 3U_{xx_2}X_{x} + 8U_{x}W_{xx_2}U + \frac{18}{5}U_{x_2}U^3V_{x} + 6V_{x}V_{xx_3} - 5UV_{xx}U_{xxx} - \frac{8}{3}UU_{xx_3}W + \frac{2}{3}U_{x}^2V_{x_3} + 2U_{xxx_3}U_{xx} - 6U_{x}^2U_{xx_3} - 3UV_{x_2}U_{x}^2 + 8UU_{x}^2U_{x_3} + 6U^2U_{x_3}U_{xx} - 8U_{xx_3}UU_{xx} + 4U_{x_3}V_{x}U^2 - 2U_{xxx}U_{x_3}U + 6U_{x_2}V_{x}U_{xx} + \frac{23}{3}U_{xxx_2}U_{x}^2 - 2W_{x_2}UU_{xx} + 3W_{x_2}U^2U_{x} - \frac{1}{2}U_{xx}U_{xxxxx} - VU_{xx}U_{xxx} - \frac{1}{2}U_{xx_4}U_{xx} + \frac{6}{5}U^3V_{x_2}U_{x} - \frac{8}{3}U_{x}WU_{x_3} - \frac{1}{2}U_{x_2}U_{xx_4} - \frac{3}{2}VU_{xxx_4} + 11U_{xx_2}VV_{x} + U_{x_2}U_{xxx_3} - 2V_{x_3}U_{xx_2} + 2U_{xx_3}V_{x_2} + W_{x_2}U_{x_3} - U_{xxx_2}U_{x_3} - U_{x_2}W_{x_3} + 8UVW_{xx_2} - \frac{7}{2}U_{x}V_{xxxx_2} - 2UU_{xx}W_{xx} - 6WW_{xx_2} - \frac{4}{3}U_{x_3}W_{x}U + 6UU_{xx}U_{x}W + UU_{xxx}U_{x}V + 6UV_{xx}U_{x}V - 29UU_{x}V_{x}U_{xx} + 16U^2U_{x}U_{xx_3} - 12UU_{x}U_{xxx_3} - 8UV_{xx_3}V + 8U^2U_{xx_3}V - 6U_{x}V^2U_{xx} + \frac{25}{3}U_{xxx_2}UU_{xx} - 6UV_{x_2}U_{x}V - 33UU_{xx_2}U_{x}V - 6U_{x_2}VUV_{x} - 4U_{x_2}U^4U_{x} - \frac{3}{2}V_{xx_4}U_{x} + 6U_{x}^3W + 11U_{xx_2}U_{x}W + 2V_{x_2}VU_{xx} - \frac{3}{2}UV_{xxx_4} - 6V_{xx_2}U^2U_{x} + 5V_{xx_2}U_{x}V + V_{xx_2}UV_{x} + 13U_{xx_2}VU_{xx} + V_{xx_2}UU_{xx} + 8UU_{xx_2}W_{x} - \frac{8}{3}U_{xx_3}V^2 + 2U_{x_2}VW_{x} - \frac{9}{2}U^3U_{xxxx_2} - 3VV_{xx_4} + 4V_{xx_3}W + 4W_{xx_3}V + 8V_{xxx_3}V + 2U^2V_{xx_4} + 3U_{xxxx_3}V - U_{xxx_4}U_{x} + 2U^2U_{xxx_4} + U_{x_4}U_{x}^2 - \frac{22}{3}U_{x}U_{x_3}V_{x} - \frac{1}{2}V_{xx_2}U_{xxx} - \frac{5}{2}U_{xx}V_{xxx_2} - 4U_{xxx_2}W_{x} - 3U_{x_2}U^2W_{x} - \frac{8}{3}U^2W_{xx_3} - 2V_{xx_2}V_{xx} + 3V_{xx_2}U_{x}^2 + UU_{xx}U_{xxxx} - 3U_{x}^3U_{x_2} - 4U_{x}U_{xx}U_{x_3} - \frac{3}{2}V_{x}U_{xx_4} - 2U^5U_{xx_2} - 6U^3V_{xxx_2} + 3V_{x}^2U_{xx} - 7U_{x}^2V_{xxx} - 2U_{x}^2U_{xxxx} + 8U_{x}^3U_{xx} - V_{x_4}U_{x}U - 27UU_{xx_2}U_{x}^2 + 2U_{x_4}VU_{x} + U_{x_4}V_{x}U + 2U_{xx_4}U_{x}U + 2U_{xx_4}UV - \frac{1}{2}U_{xxxx_4}U + 10UU_{xx_2}V_{xx} + U_{x_4}U_{xx}U + 14U_{xxx_2}U_{x}V + 8U_{xxx_2}UV_{x} + 2U_{x_2}UW_{xx} + 6U_{x_2}VV_{xx} - 2W_{x_2}U_{x}V - 2W_{x_2}UV_{x} - 6U_{x_2}U_{x}^2V + U_{x_2}UU_{xxxx} + 12U_{x}^2U^2V_{x} - 6UU_{x}V_{x}^2 - 2W_{xxxx_2}U - 6U_{x}^2V_{x}V + 6VU_{xx}W_{x} - 18UU_{x}U_{xx}^2 + V_{x}VU_{xxx} - 2U_{xxx}U^2V_{x} + \frac{2}{3}U_{xx}^2U_{x_2} + \frac{20}{3}U_{xxx}UU_{xx_2} - 6U^2U_{x_2}V_{xx} + 4U_{xx}U^2V_{xx} - 6V_{x}VV_{xx} + 3UV_{x}U_{xxxx} - 4U_{x}W_{xxx_2} + 4U_{xx}^3 - 2V_{x}^3 + 4U^3V_{xx_3} + 2U_{xxx_3}W - 3V_{xx_2}W_{x} - \frac{28}{3}U_{xx_3}UV_{x} + \frac{2}{3}V_{x_3}UU_{xx} - \frac{24}{5}U^3U_{x_3}U_{x} - 2WU_{xxxx_2} + 2U^2U_{xxxxx_2} - 4UU_{x}^2U_{xxx} - 5V_{xxx}U^2U_{x} - \frac{16}{5}U^4U_{xx_3} + 5U_{x_2}U_{x}W_{x} - \frac{3}{2}U^3U_{xx_4} + 6U_{x_2}U_{x}V_{xx} + 8U_{x}U_{xx}V_{xx} + 11U_{x}V_{x}U_{xxx} + 9U_{x}U_{xx}W_{x} + 4U_{x_2}V_{x}W - \frac{3}{2}UV_{xxxxx_2} - 2W_{xx_2}U_{xx} - \frac{9}{2}U_{xxx}U^2U_{x_2} + 8U_{x_3}UU_{x}V + \frac{36}{5}U^3U_{xx_2}V + 3UV_{xxxx_3} + 2UW_{xxx_3} - 12U_{xxx_2}U^2V + 4U^2W_{xxx_2} + \frac{24}{5}U_{xxx_2}U^4 - \frac{1}{2}U_{xxxxxx_2}U - \frac{4}{3}U_{x}VV_{x_3} + 2UV_{x_2}W_{x} + 14UU_{x}V_{xxx_2} - 6W_{xx}U^2U_{x} + 5UV_{x}V_{xxx} + 6UV_{x}W_{xx} + 4UV_{xx_2}W + 5U_{x_2}V_{x}^2 - U_{xxx_2}U_{xxx} - \frac{3}{2}U_{x_4}U_{x}U^2 + \frac{29}{3}UU_{x}U_{xxxx_2} + 4U_{x}V_{x_2}W - U_{x}^2VU_{xx} + 4U_{xx_2}VW + 4U_{xx}U_{x_2}W + UU_{xxxxx_3} - 4U^2U_{xxxx_3} + \frac{31}{3}U_{xx_2}U_{x}U_{xx} - 2W_{x_2}U_{x}^2 - 7V_{x}V_{xxx_2} + 4U_{xxx_2}UW - 2VU_{xx}V_{xx} - 6WV_{xxx_2} - 5U_{xx_2}W_{xx} + 6U^3U_{xxx_3} + 5U_{xxx_3}V_{x} + 2V_{x_2}U_{x}V_{x} + \frac{84}{5}U^3U_{xx_2}U_{x} + 4W_{x}U_{xx_3} + 2U_{x}U_{xxxx_3} + 2U_{x}W_{xx_3} + 4V_{xx}U_{xx_3} + U_{xx_3}U_{xxx} + 2V_{xx_3}U_{xx} + 5V_{xxx_3}U_{x} - 6UU_{xx_2}V^2 - \frac{5}{2}U_{xx_2}V_{xxx}- \frac{7}{2}U_{xxx_2}V_{xx} + \frac{1}{3}UV_{x_2}U_{xxx} - 3U_{xxxx}U^2U_{x} + 15U_{xx_2}U_{x}V_{x} - 6UU_{x}^2V_{xx} + \frac{48}{5}U_{x_2}U^2U_{x}V - 3U_{x}WU_{xxx} - 6U_{x}WV_{xx} - 15U_{xx}UU_{x}U_{x_2} - 15U_{x_2}UU_{x}V_{x} - 5VV_{xxxx_2} - 6VW_{xxx_2} - 8U_{x}^2W_{xx} - W_{x}V_{xxx} + 2V_{xx}W_{xx} + 2V_{xx}V_{xxx} + 3W_{xx}U_{xxx} - \frac{5}{2}U_{xxxx_2}V_{x} - \frac{16}{3}U_{x_3}VU_{xx} - \frac{20}{3}U^2V_{xxx_3} - 3U^3W_{xx_2} + 4U_{xxx_2}V^2 + \frac{1}{2}U_{x_2}V_{x_4} + 12V_{xxx_2}UV - 2UW_{x}U_{xxx}
\end{dmath*}

\begin{dmath}
 - 6UU_{xx}X_{x} + 2U_{xxx}U_{x}U^3 + 3U_{xxx}U^2U_{xx} + \frac{12}{5}U^4V_{xx_2} - \frac{39}{2}U_{xxx_2}U^2U_{x} - 6W_{xx_2}V_{x} - U_{xxxxx_2}U_{x} - \frac{38}{3}U_{x}U_{xx_3}V + \frac{66}{5}U^2U_{x}^2U_{x_2} + \frac{36}{5}U_{xx}U^3U_{x_2} - 6U^2V_{xx_2}V - 2U_{xx}W_{xxx} + 10V_{x}U_{xx}^2 + 6X_{x}V_{xx} + 3X_{x}U_{xxx} - 4V_{x}W_{xxx} - 3V_{x}V_{xxxx} - V_{x}U_{xxxxx} - 6W_{x}W_{xx} - 6U_{x}^2X_{x} + 2V_{xxx}U_{xxx} - 2UU_{xxx}^2 - 2UV_{xx}^2 + 18U_{xx}U_{x}^2U^2 - 4U_{x}^3U^3 - 2U_{xx}^2U^3 - 3U^2U_{xx_2}W - 6UU_{x}^2W_{x} + \frac{4}{3}U_{x}W_{x_3}U + \frac{7}{3}U_{x_2}UV_{xxx} + \frac{1}{2}U_{xxx}U_{xxxx} - \frac{28}{3}U_{xxx_3}UV - 4U_{x_3}VV_{x} - \frac{3}{2}VU_{xxxxx_2} - \frac{1}{2}V_{x_2}U_{x_4} - 6U_{x_2}U_{x}UW - \frac{10}{3}U_{x_3}UV_{xx} + \frac{22}{3}U_{xxxx_2}UV - \frac{22}{3}UU_{x}V_{xx_3} - \frac{3}{2}U_{xx}V_{xxxx} - \frac{27}{2}U^2U_{xx_2}U_{xx} + \frac{5}{3}U_{x_2}U_{x}U_{xxx} - \frac{1}{2}U_{xx_2}U_{xxxx} + \frac{10}{3}U_{x_2}VU_{xxx} + \frac{1}{2}V_{xx}U_{xxxx} + \frac{1}{2}U_{xx_2}U_{x_4}.
\end{dmath}
The Lagrangian $\widetilde{\mathscr{L}}_{{(234)}}$  is identical to ${\mathscr{L}}_{{(234)}}$.  From the Lagrangians given here for $1<i,j\leq 4$, we see that $\widetilde{\mathscr{L}}_{{(1ij)}}$ gives a shorter Lagrangian than $\mathscr{L}_{(1ij)}$.  In general, the difference between $\widetilde{\mathscr{L}}_{{(1ij)}}$ and ${\mathscr{L}}_{{(1ij)}}$ can be expressed as the sum of a total $x_i$ derivative and a total $x_j$ derivative.

\bibliographystyle{unsrt}

\bibliography{KP0720}

\begin{thebibliography}{10}

\bibitem{Lobb2009}
S.~Lobb and F.W. Nijhoff.
\newblock Lagrangian multiforms and multidimensional consistency.
\newblock {\em Journal of Physics A: Mathematical and Theoretical},
  42(45):454013, 2009.

\bibitem{YooKong2011}
Sikarin Yoo-Kong, Sarah Lobb, and Frank Nijhoff.
\newblock Discrete-time {C}alogero-{M}oser system and {L}agrangian 1-form
  structure.
\newblock {\em Journal of {P}hysics. A, Mathematical and {T}heoretical},
  44(36):365203--, 2011.

\bibitem{Suris2017}
M.~Petrera and Y.B. Suris.
\newblock Variational symmetries and pluri-{L}agrangian systems in classical
  mechanics.
\newblock {\em Journal of Nonlinear Mathematical Physics}, 24(sup1):121--145,
  2017.

\bibitem{Suris2016}
Y.B. Suris and M.~Vermeeren.
\newblock On the lagrangian structure of integrable hierarchies.
\newblock In A.I. Bobenko, editor, {\em Advances in Discrete Differential
  Geometry}, pages 347--378. Springer Berlin Heidelberg, Berlin, Heidelberg,
  2016.

\bibitem{Sleigh2019}
D.~Sleigh, F.W. Nijhoff, and V.~Caudrelier.
\newblock A variational approach to {L}ax representations.
\newblock {\em Journal of Geometry and Physics}, 142:66 -- 79, 2019.

\bibitem{Vermeeren2020}
M.~Petrera and M.~Vermeeren.
\newblock Variational symmetries and pluri-{L}agrangian structures for
  integrable hierarchies of {PDE}s.
\newblock {\em European Journal of Mathematics}, Nov 2020.

\bibitem{ThisPaper}
D.~{Sleigh}, F.~W. {Nijhoff}, and V.~{Caudrelier}.
\newblock {Variational symmetries and Lagrangian multiforms}.
\newblock {\em Letters in Mathematical Physics}, 110(4):805--826, Apr 2020.

\bibitem{Caud2020}
V.~Caudrelier and M.~Stoppato.
\newblock Multiform description of the {AKNS} hierarchy and classical r-matrix.
\newblock {\em Journal of Physics A: Mathematical and Theoretical},
  54(23):235204, May 2021.

\bibitem{KP1970}
B.B. {Kadomtsev} and V.I. {Petviashvili}.
\newblock On the stability of solitary waves in weakly dispersing media.
\newblock {\em Soviet Physics Doklady}, 15:539, Dec 1970.

\bibitem{Sato1981}
M.~Sato.
\newblock Soliton equations as dynamical systems on infinite dimensional
  {G}rassmann manifolds.
\newblock {\em RIMS Kokyuroku}, 439:30--46, 1981.

\bibitem{Quispel2009}
S.~B. Lobb, F.~W. Nijhoff, and G.~R.~W. Quispel.
\newblock Lagrangian multiform structure for the lattice {KP} system.
\newblock {\em Journal of Physics A: Mathematical and Theoretical},
  42(47):472002, Nov 2009.

\bibitem{Dickey1987}
L.~A. Dickey.
\newblock On {H}amiltonian and {L}agrangian formalisms for the {KP}-hierarchy
  of integrable equations.
\newblock {\em Annals of the New York Academy of Sciences}, 491(1):131--148,
  1987.

\bibitem{Vermeeren2018}
M.~Vermeeren.
\newblock {\em Continuum limits of variational systems}.
\newblock PhD thesis, Technische Universität Berlin, 2018.

\bibitem{DickeyBook}
L.~A. Dickey.
\newblock {\em Soliton Equations and Hamiltonian Systems}.
\newblock World Scientific, 2nd edition, 2003.

\bibitem{Gelfand1976}
I.~M. Gel'fand and L.~A. Dikii.
\newblock Fractional powers of operators and {H}amiltonian systems.
\newblock {\em Functional Analysis and Its Applications}, 10(4):259--273, Oct
  1976.

\bibitem{Brunelli}
J.C Brunelli.
\newblock {PSEUDO}: applications of streams and lazy evaluation to integrable
  models.
\newblock {\em Computer Physics Communications}, 163(1):22--40, 2004.

\end{thebibliography}

\end{document}